\documentclass[a4paper,10pt]{paper}
\usepackage{mathrsfs}
\usepackage{amsmath,amssymb,amsthm}
\usepackage{cite}
\usepackage{color}
\usepackage{hyperref}
\usepackage{tikz}

\newtheorem{theorem}{\bf Theorem}[section]
\newtheorem{lemma}{\bf Lemma}[section]

\newtheorem{proposition}{\bf Proposition}[section]
\numberwithin{equation}{section}

\newcommand{\bse}{\begin{subequations}}
\newcommand{\ese}{\end{subequations}}

\DeclareMathOperator{\tr}{Tr}
\def\undertilde#1{\mathord{\vtop{\ialign{##\crcr
$\hfil\displaystyle{#1}\hfil$\crcr\noalign{\kern1.5pt\nointerlineskip}
$\hfil\widetilde{}\hfil$\crcr\noalign{\kern-6.5pt}}}}}
\def\underhat#1{\mathord{\vtop{\ialign{##\crcr
$\hfil\displaystyle{#1}\hfil$\crcr\noalign{\kern1.5pt\nointerlineskip}
$\hfil\widehat{}\hfil$\crcr\noalign{\kern-6.5pt}}}}}
\def\underbar#1{\mathord{\vtop{\ialign{##\crcr
$\hfil\displaystyle{#1}\hfil$\crcr\noalign{\kern1.5pt\nointerlineskip}
$\hfil\bar{}\hfil$\crcr\noalign{\kern-6.5pt}}}}}

\def\wt#1{\widetilde{#1}}
\def\wh#1{\widehat{#1}}
\def\wb#1{\bar{#1}}
\def\th#1{\wh{\wt{#1}}}
\def\hb#1{\wb{\wh{#1}}}
\def\bt#1{\wt{\wb{#1}}}
\def\thb#1{\wb{\wh{\wt{#1}}}}
\def\ut#1{\undertilde{#1}}
\def\uh#1{\underhat{#1}}
\def\ub#1{\underbar{#1}}
\def\uth#1{\ut{\uh{#1}}}

\newcommand{\rT}{\mathrm{T}}
\newcommand{\rd}{\mathrm{d}}
\newcommand{\fL}{\mathfrak{L}}
\newcommand{\tfL}{{}^{t\!}{\mathfrak{L}}}
\newcommand{\fO}{\mathfrak{O}}
\newcommand{\fp}{\mathfrak{p}}
\newcommand{\fq}{\mathfrak{q}}

\newcommand{\bLd}{\mathbf{\Lambda}}
\newcommand{\tbLd}{{}^{t\!}\mathbf{\Lambda}}

\newcommand{\bo}{\mathbf{o}}
\newcommand{\tbo}{{}^{t\!}\mathbf{o}}
\newcommand{\bO}{\mathbf{O}}
\newcommand{\bOa}{\mathbf{\Omega}}

\newcommand{\bC}{\mathbf{C}}
\newcommand{\tbC}{{}^{t\!}\mathbf{C}}
\newcommand{\bU}{\mathbf{U}}
\newcommand{\tbU}{{}^{t\!}\mathbf{U}}
\newcommand{\bu}{\mathbf{u}}
\newcommand{\tbu}{{}^{t\!}\mathbf{u}}

\newcommand{\bv}{\mathbf{v}}
\newcommand{\tbv}{{}^{t\!}\mathbf{v}}
\newcommand{\bc}{\mathbf{c}}
\newcommand{\tbc}{{}^{t\!}\mathbf{c}}
\newcommand{\bA}{\mathbf{A}}

\newcommand{\bM}{\mathbf{M}}
\newcommand{\bI}{\mathbf{I}}
\newcommand{\bV}{\mathbf{V}}

\newcommand{\Oa}{\Omega}

\title{Direct linearising transform for three-dimensional discrete integrable systems: the lattice AKP, BKP and CKP equations
}
\author{Wei FU and Frank W NIJHOFF \\
 School of Mathematics, University of Leeds, Leeds LS2 9JT, UK 
}

\begin{document}

\maketitle

\begin{abstract}
A unified framework is presented for the solution structure of three-dimensional discrete integrable systems, including the lattice AKP, BKP and CKP equations. 
This is done through the so-called direct linearising transform which establishes a general class of integral transforms between solutions.
As a particular application, novel soliton-type solutions for the lattice CKP equation are obtained. 
\paragraph{Keywords:} direct linearising transform, discrete integrable system, 
discrete AKP equation, discrete BKP equation, discrete CKP equation, soliton 
\end{abstract}

\section{Introduction}\label{S:Intro}
Discrete integrable systems have played an increasingly prominent role in both mathematics and physics during the past decades, cf. e.g. \cite{HJN16}. 
The theory of discrete integrable systems has made important contributions to other areas of classical and modern mathematics, such as algebraic 
geometry, discrete geometry, Lie algebras, cluster algebras, orthogonal polynomials, quantum theory, special functions, 
and random matrix theory. 

Discrete equations in a sense play the role of master equations in integrable systems theory. They encode entire hierarchies of corresponding continuous 
equations and can be understood as the Bianchi permutability property and B\"acklund transforms of both continuous and discrete equations. Furthermore, 
they have quite significance in their own right, in view of the rich algebraic structure behind them. 

A key feature of the integrability of discrete systems is the phenomenon of \emph{multi-dimensional consistency} (MDC) -- the property that a lattice 
equation can be consistently extended to a family of equations by introducing an arbitrary number of  discrete independent variables with their 
corresponding lattice parameters, cf. \cite{DS97,NW01}. The MDC property was later employed to classify scalar affine-linear quadrilateral equations 
\cite{ABS03} and octahedral equations \cite{ABS12} by Adler, Bobenko and Suris. 

At the current stage, three-dimensional (3D) integrable lattice equations are often considered as the most general models in discrete integrable systems 
theory -- two-dimensional (2D) and one-dimensional (1D) discrete integrable systems can normally be obtained from dimensional reductions of 3D equations. 
In fact, the algebraic/solution structures behind 3D lattice equations are much richer and sometimes they provide us with insights into the study of 
discrete integrable systems, while on the 2D/1D level a lot of key information collapses and the integrability sometimes cannot be easily figured out. 
Under the assumption of the MDC, there are three important scalar models of Kadomtsev--Petviashvili-type (KP) in the 3D theory, namely, the lattice AKP, 
BKP and CKP equations, which are the discretisations of the famous continuous (potential) AKP, BKP and CKP equations: 
\begin{align}
 &4u_{xt}-3u_{yy}-(u_{xxx}+6u_x^2)_x=0, \tag{AKP} \\
 &9u_{xt}-5u_{yy}+(-5u_{xxy}-15u_xu_y+u_{xxxxx}+15u_xu_{xxx}+15u_x^3)_x=0, \tag{BKP} \\
 &9u_{xt}-5u_{yy}+(-5u_{xxy}-15u_xu_y+u_{xxxxx}+15u_xu_{xxx}+15u_x^3+\tfrac{45}{4}u_{xx}^2)_x=0. \tag{CKP}
\end{align}
Here the letters `A', `B' and `C' refer to the different types of infinite-dimensional Lie algebras which are associated with their respective hierarchies 
following the work of the Kyoto school, cf. a review paper \cite{JM83} and references therein for their original research papers. The lattice AKP equation 
was first given by Hirota \cite{Hir81} in the discrete bilinear form and therefore it is also known as the Hirota equation (also referred to as the 
Hirota--Miwa equation due to Miwa's reparametrisation \cite{Miw82} for its soliton solution). The lattice AKP equation is related to other nonlinear forms 
which can be referred to as the lattice KP equation \cite{NCWQ84,NCW85}, the lattice modified KP (mKP) equation \cite{DJM82b,NCW85} as well as the lattice 
Schwarzian KP (SKP) equation -- a lattice equation in the form of a multi-ratio (see \cite{DN91}). The lattice BKP equation was derived by Miwa \cite{Miw82} 
(therefore also referred to as the Miwa equation) as a four-term bilinear equation and its nonlinear form in terms of multi-ratios was later given by Nimmo 
and Schief \cite{NS97} (see also \cite{NS98}). The lattice CKP equation was obtained from the star-triangle transform in the Ising model by Kashaev 
\cite{Kas96} based on the idea of nonlocal Yang--Baxter maps \cite{MN89}. It was named CKP by Schief \cite{Sch03} who revealed that Kashaev's lattice model 
is the superposition property of the continuous CKP equation. 

We note that the MDC property of the Hirota--Miwa equation and the Miwa equation, i.e. the lattice AKP and BKP equations, can be proven by direct computation, 
however, the MDC property of the lattice CKP equation (as an equation in multi-quadratic form) is highly nontrivial and it was confirmed in \cite{TW09}. 
Alternatively, Atkinson established the MDC property of the lattice CKP equation by using discriminant factorisation \cite{Atk11}. The reductions of the 
lattice KP-type equations give rise to a large number of lower-dimensional integrable models (cf. e.g. \cite{HJN16} and references therein). In particular, 
the ultradiscretisations of the reduced 2D integrable lattice models as well as Yang--Baxter maps can be obtained from them \cite{KNW09,KNW10}. 

The lattice AKP, BKP and CKP equations have been considered from the perspective of the underlying geometry in several papers, cf. Konopelchenko and Schief 
\cite{KS02a,KS02b,Sch03}, Doliwa \cite{Dol07,Dol10a,Dol10b} and also Bobenko and Schief \cite{BS15,BS17}. In the present paper, we propose a unified framework 
for the solution structure of these equations. The latter will comprise the structure of soliton-type solutions and those related to nonlocal Riemann--Hilbert 
problems. As a particular by-product, we obtain novel soliton solutions to the lattice CKP equation. The approach we adopt is the \emph{direct linearisation} 
(DL) method, which was proven very effective in establishing solution structures of many integrable equations and their interrelations, cf. e.g. 
\cite{FA81,FA83,NQC83,NCWQ84,NCW85}. 

The starting point in the DL is a linear integral equation. In the 3D case, we need a singular nonlocal integral equation which reads
\begin{align}\label{IntEq}
 \bu_k+\iint_{D}\rd \zeta(l,l')\rho_k\Omega_{k,l'}\sigma_{l'}\bu_l=\rho_k\bc_k,
\end{align}
where $\rd \zeta(l,l')$ is a certain measure for the double integral on an integration domain $D$ in the space of the spectral variables $l$ and $l'$, 
and the wave function $\bu_k$ is an infinite vector with its entries as functions of discrete dynamical variables as well as the spectral variable 
$k$ and $\bc_k$ is also an infinite vector with its $i$th-component $k^i$. $\Omega_{k,l'}$ is called the Cauchy kernel while $\rho_k$ and $\sigma_{l'}$ 
are the plane wave factors. The key point in this approach is that the measure $\rd\zeta(l,l')$ depends on $l$ and $l'$ and therefore a double integral 
must be involved. Once the measure collapses the integral equation turns out to be a local Riemann--Hilbert problem and 2D integrable lattice equations arise. 

A more general form of the DL is the \emph{direct linearising transform} (DLT), cf. \cite{Nij85a,NC90}, which plays the role of a dressing operation 
in the framework. Concretely, the DLT is a linear transform that maps a solution of the linear problem to another and simultaneously its nonlinear 
counterpart also brings a solution of the nonlinear equation to a new solution, i.e. one can start from a seed and then generate more and more complicated 
solutions by iteration. Particularly, when the free wave is chosen to be the seed of the linear problem, the problem turns out to be the integral equation 
\eqref{IntEq} which brings us the nonlinear equation together with its solutions.

We establish a general DLT for 3D lattice equations through a structure of infinite matrices which provide a natural framework for extracting the 
relevant quantities in the theory. The kernel of the integral transform, which is defined as the sum of a lattice 1-form over a path in the lattice, 
satisfies itself an integral equation which in turn implies that the DLT for 3D lattice equations obeys a group-like property tantamount to the integrability 
of the scheme. The crucial aspect of the path-independence of the kernel is equivalent to a closure relation in our infinite matrix structure; it is this 
closure relation that forms the key condition on the construction of various 3D discrete integrable systems. As the three prominent scalar 3D integrable 
lattice equations, the lattice AKP, BKP and CKP equations emerge from the scheme together with their associated linear problems. Soliton solutions 
are obtained by specifying appropriate integration measures and domains. 

The paper is organised as follows: In Section \ref{S:DLT} we establish the DLT for general 3D discrete integrable systems and its relation to the DL. 
Sections \ref{S:AKP}, \ref{S:BKP} and \ref{S:CKP} are contributed to the lattice AKP, BKP and CKP equations which are corresponding to the three 
particular cases under the general framework of the DLT. In Section \ref{S:Soliton}, soliton solutions to the lattice AKP, BKP and CKP equations 
are given from the scheme.

\section{A general framework: Direct linearising transform}\label{S:DLT}

\subsection{Infinite matrices and vectors}
Infinite matrices and vectors are the key ingredients in the framework. A matrix/vector is normally understood as an array of numbers, symbols, etc. 
In fact, such a visualisation is not necessary for describing the essential structure of matrices and vectors, we here present an alternative 
presentation which is more suitable for the purpose of this paper. This involves the use of infinite ``centred'' matrices which are built by means of 
objects (generators) $\bLd$, $\tbLd$ and $\bO$ in an associative algebra (with unit) $\mathcal{A}$ over a field $\mathcal{F}$, obeying the relation 
\begin{align}\label{InfMatRule}
 \bO\cdot\tbLd^i\cdot\bLd^j\cdot\bO=\delta_{-i,-j}\bO, \quad i,j\in\mathbb{Z},
\end{align}
where the powers $\tbLd^i$ and $\bLd^j$ are the $i$th and $j$th compositions of $\tbLd$ and $\bLd$ respectively, and $\delta_{\cdot,\cdot}$ is the usual 
Kronecker symbol. In general, $\bLd$, $\tbLd$ and $\bO$ do not commute, but from \eqref{InfMatRule} it can be seen that $\bLd$ and $\tbLd$ act as each other's 
transpose, while $\bO$ is a projector satisfying $\bO^2=\bO$. An, in general, infinite matrix $\bU$ is defined as
\begin{align}\label{InfMat}
 \bU=\sum_{i,j\in\mathbb{Z}}U_{i,j}\bLd^{-i}\cdot\bO\cdot\tbLd^{-j},
\end{align}
where the coefficients $U_{i,j}$ take values in the field $\mathcal{F}$. The $U_{i,j}$ can be understood as the $(i,j)$-entry in the infinite matrix and 
$\{\bLd^{-i}\cdot\bO\cdot\tbLd^{-j}|i,j\in\mathbb{Z}\}$ forms a basis. In particular, $\bO$ itself can be visualised as an infinite matrix having the 
$(0,0)$-entry $1$ and all the other entries zero. Following from \eqref{InfMatRule}, we can show that operations of these infinite matrices such as addition, 
multiplication of two infinite matrices as well as scalar multiplication of an infinite matrix by $\fp\in\mathcal{F}$ obey the following rules: 
\begin{align*}
 &\bU+\bV=\sum_{i,j\in\mathbb{Z}}(U_{i,j}+V_{i,j})\bLd^{-i}\cdot\bO\cdot\tbLd^{-j}, \quad 
 \fp\,\bU=\sum_{i,j\in\mathbb{Z}}(\fp\, U_{i,j})\bLd^{-i}\cdot\bO\cdot\tbLd^{-j}, \\
 &\bU\cdot\bV=\sum_{i,j\in\mathbb{Z}}\Big(\sum_{i'\in\mathbb{Z}}U_{i,i'}V_{i',j}\Big)\bLd^{-i}\cdot\bO\cdot\tbLd^{-j}.
\end{align*}
In addition, the transpose of $\bU$ is defined as $\tbU=\sum_{i,j\in\mathbb{Z}}U_{i,j}\bLd^{-j}\cdot\bO\cdot\tbLd^{-i}$. The action $(\,\cdot\,)_{0,0}$ 
on an arbitrary infinite matrix defines the centre of the matrix $\bU$ by taking the coefficient of $\bLd^0\cdot\bO\cdot\tbLd^0$, i.e. 
$(\bLd^{i}\cdot\bU\cdot\tbLd^{j})_{0,0}=U_{i,j}$. In particular, we have the following formula: 
\begin{align*}
 (\bLd^{i_1}\cdot\bU\cdot\tbLd^{j_1}\cdot\bO\cdot\bLd^{i_2}\cdot\bU\cdot\tbLd^{j_2})_{0,0}=U_{i_1,j_1}U_{i_2,j_2}. 
\end{align*}
In fact, from the definition \eqref{InfMat}, one can also observe that $\bLd$ and $\tbLd$ play the role as index-raising operators from the left and from 
the right respectively, as we have
\begin{align*}
\bLd^{i'}\cdot\bU\cdot\tbLd^{j'}=\sum_{i,j\in\mathbb{Z}}U_{i+i',j+j'}\bLd^{-i}\cdot\bO\cdot\tbLd^{-j}.
\end{align*}
We note that the definition \eqref{InfMat} also covers finite matrices by restricting the number of non-zero coefficients to a finite number, i.e. 
$U_{i,j}=0$ for $i\neq 1,2,\cdots,N$ or $j\neq 1,2,\cdots,N'$ for some given $N$ and $N'$, resulting in an $N\times N'$ finite matrix given by 
$\bU=\sum_{i=1,j=1}^{N,N'}U_{i,j}\bLd^{-i}\cdot\bO\cdot\tbLd^{-j}$. 

Following from the same idea, we also introduce infinite vectors as follows. Suppose $\bo$ and $\tbo$ are two objects obeying the relations
\begin{align}\label{InfVecRule}
 \tbo\cdot\tbLd^i\cdot\bLd^j\cdot\bo=\delta_{-i,-j}, \quad \bo\cdot\tbo=\bO,
\end{align}
where $i,j\in\mathbb{Z}$ and objects $\bo$, $\tbo$, $\bLd$ and $\tbLd$ do not commute with each other in general. An infinite vector $\bu$ and its 
transpose $\tbu$ are defined as
\begin{align}\label{InfVec}
 \bu=\sum_{i\in\mathbb{Z}}u_i\bLd^{-i}\cdot\bo, \quad \tbu=\sum_{i\in\mathbb{Z}}u_i\tbo\cdot\tbLd^{-i},
\end{align}
where $u_i$ are elements in the same field $\mathcal{F}$, which implies $\bo$ and $\tbo$ are a particular infinite ``centred'' vector and its transpose 
that have the $0$th-component (i.e. the centre) $1$ and all the other components zero. The sets $\{\bLd^{-i}\cdot\bo|i\in\mathbb{Z}\}$ and 
$\{\tbo\cdot\tbLd^{-i}|i\in\mathbb{Z}\}$ form the respective bases for infinite vectors and their transpose. For two arbitrary infinite vectors 
$\bu=\sum_{i\in\mathbb{Z}}u_i\bLd^{-i}\cdot\bo$ and $\bv=\sum_{j\in\mathbb{Z}}v_j\bLd^{-j}\cdot\bo$ where $u_i,v_j$ are elements from the field $\mathcal{F}$ 
and for arbitrary $\fp$ also from $\mathcal{F}$, we have the basic operations for infinite vectors as follows:
\begin{align*}
 &\bu+\bv=\sum_{i\in\mathbb{Z}}(u_i+v_i)\bLd^{-i}\cdot\bo, \quad  
 \fp\,\bu=\sum_{i\in\mathbb{Z}}(\fp\, u_i)\bLd^{-i}\cdot\bo, \\
 &\bu\,\tbv=\sum_{i,j\in\mathbb{Z}}u_iv_j\bLd^{-i}\cdot\bO\cdot\tbLd^{-j}, \quad 
 \tbv\,\bu=\sum_{i\in\mathbb{Z}}u_iv_i,
\end{align*}
as they can be derived immediately by considering the rules in \eqref{InfVecRule}. The operations obey the same rules (except $i,j$ are now chosen from 
$\mathbb{Z}$) as those for finite vectors. In fact, the case of $N$-component (column and row) vectors is obtained by the restriction $u_i=0$ for 
$i\neq 1,2,\cdots,N$, leading to $\bu=\sum_{i=1}^N u_i\bLd^{-i}\cdot\bo$ and $\tbu=\sum_{i}^Nu_i\tbo\cdot\tbLd^{-i}$.
For infinite vectors, we define the action $(\,\cdot\,)_0$ by taking the coefficient of $\bLd^0\cdot\bo$ and of $\tbo\cdot\tbLd^0$ for $\bu$ and $\tbu$ 
respectively, namely $(\bLd^{i}\cdot\bu)_0=u_i$ and $(\tbu\cdot\tbLd^{j})_0=u_j$. The action of $\bLd$ and $\tbLd$ as index-raising operators can be most 
clearly seen by introducing monomial eigenvectors $\bc_k$ and $\tbc_{k'}$ defined by 
\begin{align}\label{ck}
 \quad \bLd\cdot\bc_k=k\bc_k, \quad \tbc_{k'}\cdot\tbLd=k'\tbc_{k'} \quad \hbox{and} \quad (\bc_k)_0=1, \quad (\tbc_{k'})_0=1,
\end{align}
from which it implies that the $i$th-components for the eigenvectors are $k^i$ and $k'^i$, respectively. The monomial basis allows us to encode $k$- and 
$k'$-dependent objects that occur in the integral equation \eqref{IntEq} in terms of infinite matrices and vectors and work with them conveniently in 
a formal manner. 

Following the rules \eqref{InfMatRule} and \eqref{InfVecRule}, the multiplication of an infinite matrix and an infinite vector is also well-defined and 
it obeys exactly the same rule as that in the theory of finite matrices and vectors. Some examples involving both infinite matrices and vectors are 
given below: 
\begin{align*}
 (\bLd^{i_1}\cdot\bU\cdot\tbLd^{j_1}\cdot\bO\cdot\bLd^{i_2}\cdot\bu)_{0,0}=U_{i_1,j_1}u_{i_2}, \quad 
 \tbc_{k'}\cdot\tbLd^{j}\cdot\bO\cdot\bLd^{i}\cdot\bc_k=k^{i}k'^{j}.
\end{align*}
In the approach we also need the notion of trace and determinant of an infinite matrix, but only in the case when they involve the projector $\bO$. 
Thus we have for the trace $\tr(\bU\cdot\bO)=\tr(\bO\cdot\bU)=(\bU)_{0,0}$ and for the determinant 
\begin{align}\label{WARank1}
 \det(1+\bU\cdot\bO\cdot\bV)=\det(1+\bU\cdot\bo\cdot\tbo\cdot\bV)=1+\tbo\cdot\bV\cdot\bU\cdot\bo=1+(\bV\cdot\bU)_{0,0},
\end{align}
where we treat $\bU\cdot\bo$ and $\tbo\cdot\bV$ as an infinite column vector and an infinite row vector respectively and \eqref{InfMatRule} and 
\eqref{InfVecRule} are used. This identity is a version of the well-known Weinstein--Aronszajn formula in rank 1 case, evaluating the determinant 
of an infinite matrix in terms of a scalar quantity.

\subsection{Direct linearising transform for 3D lattice equations}

We consider an infinite-dimensional regular lattice spanned by discrete coordinates (called the lattice variables) $n_1,n_2,n_3,\cdots$ (for the 
sake of this paper, taking integer values), each of which is associated with a corresponding lattice parameter $p_\gamma$ taking values in the reals 
or the complex numbers (we can think of $p_\gamma$ representing a grid width parameter in the $n_\gamma$-direction of the lattice). A discrete shift 
operator is defined as follow. Suppose $f=f(n_1,n_2,n_3,\cdots)$ is a function of the discrete independent variables $\{n_1,n_2,n_3,\cdots\}$, 
the discrete forward shift operator $\rT_{p_\gamma}$ and the discrete backward shift operator $\rT_{p_\gamma}^{-1}$ are respectively defined by 
\begin{align*}
 \rT_{p_\gamma} f\doteq f(n_1,n_2,\cdots,n_\gamma+1,\cdots), \quad \rT_{p_\gamma}^{-1} f\doteq f(n_1,n_2,\cdots,n_\gamma-1,\cdots).
\end{align*}
For a 3D discrete integrable system, we only need three discrete independent variables. Thus we can choose 
an arbitrary triplet $(n_\alpha,n_\beta,n_\gamma)$ from $\{n_1,n_2,n_3,\cdots\}$ as the discrete variables and $(p_\alpha,p_\beta,p_\gamma)$ is the 
corresponding triplet for the associated lattice parameters. 

We consider a family of linear equations taking the form of 
\begin{align}\label{ukDyn}
 \rT_\fp\bu_k=(\fL_\fp-\rT_\fp\bU\cdot\fO_p)\cdot\bu_k, \quad \tbu_{k'}=\rT_\fp\tbu_{k'}\cdot(\tfL_\fp+\fO_\fp\cdot\bU), \quad
 \fp\in\{p_\gamma|\gamma\in\mathbb{Z}^+\},
\end{align}
where the wave functions $\bu_k$ and $\tbu_{k'}$ are an infinite column vector and an infinite row vector respectively composed of corresponding entries 
$u_k^{(i)}$ and $u_{k'}^{(j)}$ as functions of the discrete independent variables $\{n_1,n_2,n_3,\cdots\}$ and of certain spectral parameters $k$ and $k'$ 
which take values in the complex numbers (we prefer to denote the dependence on $k$ and $k'$ in $\bu_k$, $\tbu_{k'}$ by suffices in order to make it 
more visible). The potential $\bU$ is an infinite matrix with entries $U_{i,j}$ as certain functions of only the discrete independent variables 
$\{n_1,n_2,n_3,\cdots\}$. The operators $\fL_\fp$ and $\tfL_\fp$ are expressions of $\bLd$ and $\fp$ and $\tbLd$ and $\fp$ respectively and $\fO_\fp$ is 
an expression of the operators $\bLd$, $\tbLd$ and $\bO$ associated with the lattice parameter $\fp$. We note that while the $\fL_\fp$ commute among 
themselves for arbitrary $\fp$, and so do the $\tfL_\fp$, they do not commute with the $\fO_\fp$. This follows from the commutativity and noncommutativity 
of $\bLd$, $\tbLd$ and $\bO$. 

It is not obvious that \eqref{ukDyn} constitutes a compatible system of linear equations, i.e. they can be imposed simultaneously on one and the same 
function $\bu_k$ or $\tbu_{k'}$ respectively. The compatibility of the system depends sensitively on the choice of $\fL_\fp$, $\tfL_\fp$ and $\fO_\fp$. 
The aim is to find choices of these operators for which the compatibility holds and the framework that we set up in this subsection gives necessary 
conditions for selecting such operators. Furthermore, the requirement of compatibility also imposes conditions on the coefficient matrix $\bU$, namely 
for each triplet $(p_\alpha,p_\beta,p_\gamma)$, the compatibility conditions $\rT_{p_\alpha}\rT_{p_\beta}\bu_k=\rT_{p_\beta}\rT_{p_\alpha}\bu_k$ and 
$\rT_{p_\beta}\rT_{p_\gamma}\bu_k=\rT_{p_\gamma}\rT_{p_\beta}\bu_k$ leads to a formal 3D lattice system for the infinite matrix $\bU$ with entries as 
functions of the triplet $(n_\alpha,n_\beta,n_\gamma)$ chosen from the $\{n_1,n_2,n_3,\cdots\}$: 
\bse\label{MatEq}
\begin{align}
 &(\fL_{p_\alpha}-\rT_{p_\beta}\rT_{p_\alpha}\bU\cdot\fO_{p_\alpha})\cdot(\fL_{p_\beta}-\rT_{p_\beta}\bU\cdot\fO_{p_\beta})
 =(\fL_{p_\beta}-\rT_{p_\alpha}\rT_{p_\beta}\bU\cdot\fO_{p_\beta})\cdot(\fL_{p_\alpha}-\rT_{p_\alpha}\bU\cdot\fO_{p_\alpha}), \\
 &(\fL_{p_\beta}-\rT_{p_\gamma}\rT_{p_\beta}\bU\cdot\fO_{p_\beta})\cdot(\fL_{p_\gamma}-\rT_{p_\gamma}\bU\cdot\fO_{p_\gamma})
 =(\fL_{p_\gamma}-\rT_{p_\beta}\rT_{p_\gamma}\bU\cdot\fO_{p_\gamma})\cdot(\fL_{p_\beta}-\rT_{p_\beta}\bU\cdot\fO_{p_\beta}).
\end{align}
\ese
At this juncture we need to point out that \eqref{MatEq} can only be considered as a closed-form lattice equation in the sense of the infinite matrix $\bU$, 
but it is not closed-form for individual entries in the matrix, since the operators $\fL_\fp$, $\tfL_\fp$ and $\fO_{\fp}$ in general contain the index-raising 
operators $\bLd$ and $\tbLd$. Thus, in order to obtain closed-form 3D lattice equations for the entries of $\bU$ the latter operators must be eliminated. 
In Sections \ref{S:AKP}, \ref{S:BKP} and \ref{S:CKP}, we specify the operators $\fL_\fp$, $\tfL_\fp$ and $\fO_{\fp}$ and derive the corresponding closed-form 
lattice equations. What we do here is to set up the solution structure behind these equations which will make transparent the integrability in the sense of 
the MDC property. 

Let us explain what we mean by the MDC property in this context. Consider an, in principle infinite, family of equations arising from \eqref{MatEq}, each 
involving a triplet $(p_\alpha,p_\beta,p_\gamma)$ as lattice parameters and an associated triplet $(n_\alpha,n_\beta,n_\gamma)$ as lattice variables 
arbitrarily chosen from the infinite sets of lattice parameters and variables $\{p_\gamma|\gamma\in\mathbb{Z}^+\}$ and $\{n_\gamma|\gamma\in\mathbb{Z}^+\}$, 
the MDC is the property that this family admits a nontrivial common solution involving an arbitrary number of free parameters (different from the lattice 
parameters). In the following, we show that the MDC property is only guaranteed when $\fL_\fp$, $\tfL_\fp$ and $\fO_\fp$ satisfy certain conditions which 
are governed by the DLT. 

We introduce the DLT, namely, a linear transform from the given wave functions $\bu_k^0,\tbu_{k'}^0$ to new wave functions $\bu_k,\tbu_{k'}$ in the form of 
\begin{align}\label{ukDLT}
 &\bu_k=\bu_k^0-\iint_D\rd\zeta(l,l')G_{k,l'}^0\bu_l, \quad \tbu_{k'}=\tbu_{k'}^0-\iint_D\rd\zeta(l,l')G_{l,k'}\tbu_{l'}^0,
\end{align}
in which $D$ is a suitable integration domain and $\rd \zeta(l,l')$ is the measure. $G_{k,l'}^0$ and $G_{l,k'}$ are the kernel of the DLT that must 
obey a certain condition which will be given later. The idea of the DLT is that the transform must preserve the structure of the linear equations 
in \eqref{ukDyn}, i.e. if we start from an old solution $\bu_k^0$, $\tbu_{k'}^0$ to \eqref{ukDyn} together with $\bU^0$, the new $\bu_k$ and 
$\tbu_{k'}$ that follow from the linearising transform \eqref{ukDLT} must satisfy the same linear problem associated with a new $\bU$. Such a requirement 
leads that the kernel $G$ must obey 
\begin{align}\label{GDyn}
 (\rT_\fp-1)G_{k,l'}^0=(\rT_\fp\tbu_{l'}^0)\cdot\fO_\fp\cdot\bu_k^0, \quad (\rT_\fp-1)G_{l,k'}=(\rT_\fp\tbu_{k'})\cdot\fO_\fp\cdot\bu_l
\end{align}
for any $\fp\in\{p_\gamma|\gamma\in\mathbb{Z}^+\}$, which implies that it is defined as 
\begin{align}\label{G}
 G_{k,l'}^0=\sum_{\Gamma_{P\rightarrow X}}(\rT_{\fp}\tbu_{l'}^0)\cdot\fO_{\fp}\cdot\bu_k^0, \quad
 G_{l,k'}=\sum_{\Gamma_{P\rightarrow X}}(\rT_{\fp}\tbu_{k'})\cdot\fO_{\fp}\cdot\bu_l,
\end{align}
where $\Gamma$ is an arbitrary path from the initial point $P$ to the ending point $X$ in the lattice and the path must follow the lattice directions 
(cf. Figure \ref{F:Path} as a particular case when the path is sitting in a 3D lattice spanned by the coordinates $n_\alpha$, $n_\beta$ and $n_\gamma$). 
Simultaneously there is also a nonlinear counterpart of the DLT \eqref{ukDLT} that must hold, namely, a transform between the old $\bU^0$ and the new $\bU$, 
which is given by 
\begin{align}\label{UDLT}
 \bU-\bU^0=\iint_D\rd\zeta(l,l')\bu_l\tbu_{l'}^0.
\end{align}
The transform \eqref{UDLT} together with the equations in \eqref{ukDyn} also implies that $\bU$ must obey
\begin{align}\label{UU0Dyn}
 (\rT_\fp\bU-\rT_\fp\bU^0)\cdot(\tfL_\fp+\fO_\fp\cdot\bU^0)=\fL_\fp\cdot(\bU-\bU^0)-(\rT_\fp\bU)\cdot\fO_\fp\cdot(\bU-\bU^0).
\end{align}
The relation \eqref{UU0Dyn} is a nonlinearised version of \eqref{ukDyn} and it describes how the nonlinear quantities 
$\bU$ and $\bU^0$ evolve along the lattice directions. And it can also bring us closed-form 3D lattice systems by selecting a certain entry 
(for scalar 3D lattice equations) or several certain entries (for coupled 3D lattice systems) of the infinite matrix $\bU$ once the initial 
condition $\bU^0$ is given -- the simplest choice is $\bU^0=\mathbf{0}$ as we can easily observe that from \eqref{MatEq}. 

\begin{figure}[!h]
\centering
\begin{tikzpicture}[scale=0.15]
 \node at (0,-3) {$n_\alpha$};
 \node at (17.5,4.5) {$n_\beta$};
 \node at (5.5,17) {$n_\gamma$};
 \draw[very thick,->] (3,3) -- (-3,-3);
 \draw[very thick,->] (3,3) -- (18,3);
 \draw[very thick,->] (3,3) -- (3,18);
 \coordinate (A1) at (0, 0);
 \coordinate (A2) at (10, 0);
 \coordinate (A3) at (13, 3);
 \coordinate (A4) at (3, 3);
 \coordinate (B1) at (0, 10);
 \coordinate (B2) at (10, 10);
 \coordinate (B3) at (13, 13);
 \coordinate (B4) at (3, 13);
 \draw[very thick] (A1) -- (A2);
 \draw[very thick] (A2) -- (A3);
 \draw[dashed,very thick] (A3) -- (A4);
 \draw[dashed,very thick] (A4) -- (A1);
 \draw[very thick] (B1) -- (B2);
 \draw[very thick] (B2) -- (B3);
 \draw[very thick] (B3) -- (B4);
 \draw[very thick] (B4) -- (B1);
 \draw[very thick] (A1) -- (B1);
 \draw[very thick] (A2) -- (B2);
 \draw[very thick] (A3) -- (B3);
 \draw[dashed,very thick] (A4) -- (B4);
 \foreach \x in {0,2,...,10}
 \foreach \y in {0,2,...,10}
  {
  \draw[dashed] (0+0.3*\y,\x+0.3*\y) -- (10+0.3*\y,\x+0.3*\y);
  \draw[dashed] (0+\x,0+\y) -- (3+\x,3+\y);
  \draw[dashed] (0+\x+0.3*\y,0+0.3*\y) -- (0+0.3*\y+\x,10+0.3*\y);
  }
 \draw[red,very thick] (3,3) -- (2.4,2.4) -- (4.4,2.4) -- (4.4,5.4) -- (3.8,4.8) -- (5.8,4.8) -- (5.8,6.8) -- (5.2,6.2) -- (7.2,6.2) --
 (7.2,7.2) -- (6,6) -- (8,6) -- (8,9) -- (10,9) -- (10,10);
 \fill [red] (3,3) circle (10pt);
 \fill [red] (10,10) circle (10pt);
 \node [red] at (1,4) {$P$};
 \node [red] at (11.5,9) {$X$};
\end{tikzpicture}
\caption{An arbitrary path $\Gamma$ in the 3D lattice space}
\label{F:Path}
\end{figure}
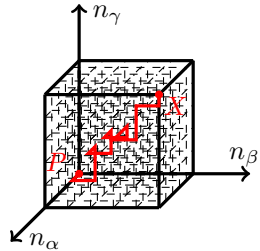

The equations in \eqref{ukDLT} together with \eqref{UDLT} constitute a dressing method for solving the 3D lattice nonlinear equation \eqref{MatEq}: 
We can start from a seed of the nonlinear equation, say $\bU^0$, then \eqref{ukDyn} can be solved and we obtain its seed solution $\bu_{k}^0$ and 
$\tbu_{k'}^0$. Next the kernel $G_{k,l'}^0$ can be calculated from \eqref{G}. Then the linearising transform \eqref{ukDLT} gives us the new solution 
$\bu_k$ to the linear problem and consequently the new solution $\bU$ to the nonlinear equation can be worked out from \eqref{UDLT}. We can repeat 
the procedure from the obtained new solution and create more and more complicated solutions by iteration. The iteration can be illustrated in the 
diagram as follow: 
\begin{align}\label{Dress}
\begin{array}{ccccccc}
 \bU^0 & {} & \bU^1 & {} & \bU^2 & {} & {} \\
 \downarrow & {} & \uparrow & {} & \uparrow & {} & {} \\
 \bu_k^0,\tbu_{k'}^0 & \longrightarrow & \bu_k^1 & \longrightarrow & \bu_k^2 & \longrightarrow & \cdots.
\end{array}
\end{align}

The DLT procedure including \eqref{ukDLT} and \eqref{UDLT} provides a way of generating solutions to \eqref{MatEq}, starting from an initial solution. 
In order to find nontrivial solutions admitting infinitely many degrees of freedom which can solve all the discrete equations for $\bU$ in the family 
(i.e. the MDC property), we must have that the kernel $G$ in the DLT given by \eqref{G} is path-independent, namely a closure relation at each stage in 
the dressing: 
\begin{align*}
 (\rT_\fq-1)(\rT_\fp-1)G_{k,l'}=(\rT_\fp-1)(\rT_\fq-1)G_{k,l'}, \quad \forall \fp,\fq\in\{p_\gamma|\gamma\in\mathbb{Z}^+\}.
\end{align*}
Following from \eqref{GDyn} and making use of \eqref{ukDyn}, one can from the above relation derive a closure relation for $\fL_\cdot$, $\tfL_\cdot$ and 
$\fO_\cdot$ by eliminating $\rT_{\fp}\bu_k$ and $\rT_{\fq}\bu_k$ (but reserve $\rT_{\fp}\rT_{\fq}\bu_k$ and $\bu_k$). The statement can be made as follow: 
\begin{proposition}\label{T:Closure}
A 3D lattice equation from \eqref{MatEq} possesses the MDC property if the operators $\fL_\cdot$, $\tfL_\cdot$ and $\fO_\cdot$ obey the closure relation 
\begin{align}\label{Closure}
\fO_\fp\fL_\fq-\tfL_\fq\fO_\fp=\fO_\fq\fL_\fp-\tfL_\fp\fO_\fq, \quad \forall \fp,\fq\in\{p_\gamma|\gamma\in\mathbb{Z}^+\}.
\end{align}
\end{proposition}
The statement holds because the closure relation \eqref{Closure} is equivalent to the path-independence of the kernel $G$, namely the DLT preserves the 
structure of the family of the linear equations \eqref{ukDyn} as well as the family of the nonlinear equations \eqref{MatEq}. In other words, a general 
solution solving the hierarchy of the discrete equations in $\bU$ can be obtained from the DLT if we start from a common seed solution, i.e. the zero 
solution $\bU^0=\mathbf{0}$, where an arbitrary number of degrees of freedom can be introduced (from the measure) during the DLT iteration \eqref{UDLT}, 
cf. Section \ref{S:Soliton} for soliton solutions. 

Another comment is that Equation \eqref{Closure} needs to be understood as an equation for the triplet $(\fL_\cdot,\tfL_\cdot,\fO_\cdot)$. Since it is 
not fully determined, it might be difficult to find its general solution. Nevertheless, one needs to keep the commutativity and noncommutativity properties 
of the operators in mind and some particular solutions can be expected from the closure relation. In fact, any solution of the closure relation \eqref{Closure} 
can result in a 3D integrable lattice equation and we shall see in the following sections that even some simple solutions can generate highly nontrivial 3D 
discrete integrable systems.

We can see from the dressing procedure that the kernel $G$ of the DLT is the key quantity in the approach. In fact, it possesses a very important 
property, namely, the kernel itself obeys a closed relation. We propose the statement in the following lemma: 
\begin{lemma}\label{L:G}
The kernel $G$ of the DLT obeys itself an integral equation
\begin{align}\label{GIntEq}
 G_{k,k'}=G_{k,k'}^0-\iint_D\rd\zeta(l,l')G_{l,k'}G_{k,l'}^0.
\end{align}
\end{lemma}

\begin{proof}
We can expand $G_{k,k'}$ following the definition \eqref{G} and replace $\bu_k$ and $\tbu_{k'}$ by $\bu_k^0$ and $\tbu_{k'}^0$ respectively according 
to the DLT \eqref{ukDLT}. Thus one obtains 
\begin{align*}
 G_{k,k'}(X)={}&G_{k,k'}^0(X)-\iint_D\rd\zeta(l,l')\sum_{\Gamma_{P\rightarrow X}}
 \Big(\sum_{\Gamma_{P\rightarrow X'}}\rT_\fq\tbu_{k'}(X'')\cdot\fO_\fq\cdot\bu_l(X'')\Big)\rT_\fp\tbu_{l'}^0(X')\cdot\fO_\fp\cdot\bu_k^0(X') \\
 &-\iint_D\rd\zeta(l,l')\sum_{\Gamma_{P\rightarrow X}}
 \Big(\sum_{\Gamma_{P\rightarrow X'}}\rT_\fq\tbu_{l'}^0(X'')\cdot\fO_\fq\cdot\bu_k^0(X')\Big)\rT_\fp\tbu_{k'}(X')\cdot\fO_\fp\cdot\bu_l(X'),
\end{align*}
where $X$ denotes the argument $\{n_1,n_2,n_3,\cdots\}$, $X'$, $X''$ denote summed over positions along the paths $\Gamma_{P\rightarrow X}$ and 
$\Gamma_{P\rightarrow X'}$ respectively, and $\rT_\cdot$ is the elementary discrete shift associated with the lattice parameters $\fp$ and $\fq$ 
governed by the paths. The identity that one needs for the proof is 
\begin{align}\label{DisSum}
 &\sum_{\Gamma_{P\rightarrow X}}\Big(\sum_{\Gamma_{P\rightarrow X'}}A(X''+e_{p''},X'')\Big)B(X'+e_{p'},X')
 +\sum_{\Gamma_{P\rightarrow X}}A(X'+e_{p'},X')\Big(\sum_{\Gamma_{P\rightarrow X'}}B(X''+e_{p'},X'')\Big) \nonumber \\
 &=\Big(\sum_{\Gamma_{P\rightarrow X}}A(X'+e_{p'},X')\Big)\Big(\sum_{\Gamma_{P\rightarrow X}}B(X''+e_{p''},X'')\Big),
\end{align}
where $e_{p'}$ and $e_{p''}$ are elementary discrete shifts and $A$, $B$ are functions defined on the infinite-dimensional lattice. We note that the 
identity is nothing but a discrete version of integration by part. Now the double summation in the above equation of $G$ can be reformulated as a product 
of two summations and the right hand side of the equation turns out to be 
\begin{align*}
 G_{k,k'}^0(X)-\iint_D\rd\zeta(l,l')
 \Big(\sum_{\Gamma_{P\rightarrow X}}\rT_\fq\tbu_{l'}^0(X'')\cdot\fO_\fq\cdot\bu_k^0(X'')\Big)
 \Big(\sum_{\Gamma_{P\rightarrow X}}\rT_\fp\tbu_{k'}(X')\cdot\fO_\fp\cdot\bu_l(X')\Big),
\end{align*}
which implies the integral equation of $G$ holds if one follows the definition \eqref{G}.
\end{proof}

With the help of \eqref{GIntEq}, one can actually establish a group-like property for the general DLT. We propose it in the following theorem: 
\begin{theorem}\label{T:Group}
The DLT for a 3D lattice equation, i.e. \eqref{ukDLT} and \eqref{UDLT}, satisfies a group-like property (see Figure \ref{F:Group}) in the sense that 
for two arbitrary DLTs 
\begin{align*}
 (\bU^0,\bu_k^0,\tbu_{k'}^0)\longrightarrow(\bU^1,\bu_k^1,\tbu_{k'}^1) \quad \hbox{and} \quad
 (\bU^1,\bu_k^1,\tbu_{k'}^1)\longrightarrow(\bU^2,\bu_k^2,\tbu_{k'}^2)
\end{align*}
with respect to the domains and the measures $(D_{10},\rd\zeta_{10}(l,l'))$ and $(D_{21},\rd\zeta_{21}(l,l'))$ respectively, the transform 
between $(\bU^0,\bu_k^0,\tbu_{k'}^0)$ and $(\bU^2,\bu_k^2,\tbu_{k'}^2)$ follows 
\begin{align*}
 &\bu_k^2=\bu_k^0-\iint_{D_{20}}\rd\zeta_{20}(l,l')G_{k,l'}^0\bu_l^2, \quad \tbu_{k'}^2=\tbu_{k'}^0-\iint_{D_{20}}\rd\zeta_{20}(l,l')G_{l,k'}^2\tbu_{l'}^0, \\
 &\bU^2-\bU^0=\iint_{D_{20}}\rd\zeta_{20}(l,l')\bu_l^2\tbu_{l'}^0, \quad \hbox{where} \quad
 \iint_{D_{20}}\,\cdot\,\rd\zeta_{20}(l,l')\doteq\iint_{D_{21}}\,\cdot\,\rd\zeta_{21}(l,l')+\iint_{D_{10}}\,\cdot\,\rd\zeta_{10}(l,l').
\end{align*}
\end{theorem}
\begin{proof}
One just needs to eliminate $\bu_k^1$ and $\tbu_{k'}^1$ in the dressings ``$0 \rightarrow 1$'' and ``$1\rightarrow 2$'' and simultaneously 
replaces $G^1$ by $G^0$ (or $G^2$), using \eqref{ukDLT} as well as the integral equation \eqref{GIntEq}.
\end{proof}
\begin{figure}[!h]
\centering
\begin{tikzpicture}[scale=0.8]
 \node at (1,0.3) {\scriptsize $D_{10},\rd\zeta_{10}(l,l')$};
 \node at (6,0.3) {\scriptsize $D_{21},\rd\zeta_{21}(l,l')$};
 \node at (-1.5,0) {\scriptsize $(\bu_k^0,\bu_{k'}^0,\bU^0)$};
 \node at (3.5,0) {\scriptsize $(\bu_k^1,\tbu_{k'}^1,\bU^1)$};
 \node at (8.5,0) {\scriptsize $(\bu_k^2,\tbu_{k'}^2,\bU^2)$};
 \node at (3.5,-1.7) {\scriptsize $D_{20},\rd\zeta_{20}(l,l')$};
 \draw [thick,->] (-0.4,0) -- (2.4,0);
 \draw [thick,->] (4.6,0) -- (7.4,0);
 \draw [thick,->] (-1.5,-0.5) to [out=-30,in=-150] (8.5,-0.5);
\end{tikzpicture}
\caption{Group-like property of the DLT}
\label{F:Group}
\end{figure}
Theorem \ref{T:Group} brings to light the underlying integrability of the structure. In fact, it implies that solutions from successive applications of 
the DLT are compatible. 

Furthermore, we note that sometimes the $\tau$-function is considered as a fundamental quantity which describes the algebraic structure behind an equation 
in the integrable systems theory, and it can be introduced within the framework of the DLT \cite{Nij85b}. However, in the infinite matrix structure of the DL, 
the $\tau$-function can be defined in a direct way. For this purpose, in the next subsection, we recover the DL from the DLT.

\subsection{Infinite matrix structure of the direct linearisation}

To recover the DL from the DLT,  the wave functions must be chosen to be the free waves. In fact, we need a seed solution to the nonlinear equation, i.e. 
the zero solution $\bU^0=\bf 0$. Then the free wave solution $\bu_k^0$, $\tbu_{k'}^0$ can be solved from the original linear problem \eqref{ukDyn} and they 
take the form of $\bu_k^0=\rho_k\bc_k$ and $\tbu_{k'}^0=\tbc_{k'}\sigma_{k'}$ respectively, where $\bc_k$ and $\tbc_{k'}$ are the monomial eigenvectors defined 
in \eqref{ck}, and the plane wave factors $\rho_k$ and $\sigma_{k'}$, as functions dependent on the lattice variables and parameters $n_\gamma$ and $p_\gamma$, 
and also the spectral parameters $k$ and $k'$ respectively, are defined by 
\begin{align}\label{PWF}
 (\rT_\fp\rho_k)\bc_k=\rho_k\fL_\fp\cdot\bc_k, \quad \tbc_{k'}\sigma_{k'}=\tbc_{k'}\cdot\tfL_\fp(\rT_\fp\sigma_{k'}), \quad 
 \fp\in\{p_\gamma|\gamma\in\mathbb{Z}^+\}.
\end{align}
Now we define the operator $\bOa$, an expression of $\bLd$, $\tbLd$ and $\bO$ which is independent of the lattice parameter $\fp$, by the following: 
\begin{align}\label{OaDyn}
 \bOa\cdot\fL_\fp-\tfL_\fp\cdot\bOa=\fO_\fp, \quad \fp\in\{p_\gamma|\gamma\in\mathbb{Z}^+\}.
\end{align}
The Cauchy kernel can immediately be recovered from the operator $\bOa$ as follow: 
\begin{align}\label{Kernel}
 \Oa_{k,l'}=\tbc_{l'}\cdot\bOa\cdot\bc_k.
\end{align}
From the definition of the Cauchy kernel \eqref{Kernel} it is obvious to see that the operator $\bOa$ is the infinite matrix form of the Cauchy kernel. 
However, we note that the infinite matrix form is much more convenient for us to deal with the computation in practice. With the free waves and the 
Cauchy kernel, we can now explain the kernel $G_{k,l'}^0$ of the DLT in this special case. In fact, making use of \eqref{OaDyn} and \eqref{Kernel} one 
can show that 
\begin{align}\label{G-tau}
 G_{k,l'}^0&=\sum_{\Gamma_{P\rightarrow X}}(\rT_\fp\sigma_{l'})\tbc_{l'}\cdot\fO_\fp\cdot\bc_k\rho_k
 =\sum_{\Gamma_{P\rightarrow X}}(\rT_\fp\sigma_{l'})\tbc_{l'}\cdot(\bOa\cdot\fL_\fp-\tfL_\fp\cdot\bOa)\cdot\bc_k\rho_k \nonumber \\
 &=\sum_{\Gamma_{P\rightarrow X}}((\rT_\fp\sigma_{l'})\Oa_{k,l'}(\rT_\fp\rho_k)-\sigma_{l'}\Oa_{k,l'}\rho_k)
 =\rho_k\Oa_{k,l'}\sigma_{l'}\Big|_P^X=\rho_k\Oa_{k,l'}\sigma_{l'},
\end{align}
in which the initial point $P$ is set to be $-\infty$ (for the sake of soliton-type solutions) in the summation along the path in order to get rid of the 
constant of summation. This tells us that $G_{k,l'}^0$ is a quantity containing both the Cauchy kernel and the plane wave factors. More precisely, the 
special case brings the DLT \eqref{ukDLT} back to the integral equation \eqref{IntEq}. With the help of the infinite matrix structure, it also makes it 
possible for us to rewrite the integral equation \eqref{IntEq} in a more compact form, namely, 
\begin{align}\label{uk}
 \bu_k=(1-\bU\cdot\bOa)\cdot\rho_k\bc_k,
\end{align}
where $\bU=\iint_D\rd\zeta(l,l')\bu_l\tbu_{l'}^0$ as it follows from \eqref{UDLT} when $\bU^0=0$. The linear equation \eqref{uk} also leads to  
a similar relation for the infinite matrix $\bU$: 
\begin{align}\label{U}
 \bU=(1-\bU\cdot\bOa)\cdot\bC,
\end{align}
where the infinite matrix $\bC$ is given by 
\begin{align}\label{C}
 \bC=\iint_D\rd\zeta(l,l')\bu_l^0\tbu_{l'}^0=\iint_D\rd\zeta(l,l')\rho_l\bc_l\tbc_{l'}\sigma_{l'}.
\end{align}
The quantity $\bC$ is an analogue of the free wave functions in the infinite matrix structure, which can be understood as a linear version of $\bU$. 
For arbitrary $\fp\in\{p_\gamma|\gamma\in\mathbb{Z}^+\}$, the infinite matrix $\bC$ satisfies the dynamical relation 
\begin{align}\label{CDyn}
 (\rT_\fp\bC)\cdot\tfL_\fp=\fL_\fp\cdot\bC.
\end{align}
Its nonlinear counterpart, i.e. the dynamical relation for the infinite matrix $\bU$, can naturally be obtained from the transform of $\bU$ \eqref{UDLT} 
in the particular case of $\bU^0=\mathbf{0}$, which now takes the form of 
\begin{align}\label{UDyn}
 (\rT_\fp\bU)\cdot\tfL_\fp=\fL_\fp\cdot\bU-(\rT_\fp\bU)\cdot\fO_\fp\cdot\bU, \quad \fp\in\{p_\gamma|\gamma\in\mathbb{Z}^+\}.
\end{align}
Alternatively, such relations can also be derived from \eqref{U} together with \eqref{OaDyn} and \eqref{CDyn}. 

Equation \eqref{UDyn} is the fundamental relation for constructing closed-form scalar 3D lattice equations within the DL framework and it is fully 
characterised by the operators $\fL_\fp$, $\tfL_\fp$ as well as $\fO_\fp$. The operators are determined by the closure relation \eqref{Closure} as we 
have pointed out in the previous subsection. In the next three sections we consider several particular solutions to the closure relation and show that 
they surprisingly lead to the well-known lattice AKP, BKP and CKP equations.

\section{The lattice AKP equation}\label{S:AKP}

As it has been shown that only three discrete independent variables are sufficient to generate a 3D lattice equation from the DLT approach and the obtained 
lattice equation possesses the MDC property, from now on we fix on the discrete independent variables $n\doteq n_1,m\doteq n_2,h\doteq n_3$ and the associated 
lattice parameters $p\doteq p_1,q\doteq p_2,r\doteq p_3$ without losing generality. Some short-hand notations can be introduced for simplicity as follows. 
Suppose $f\doteq f_{n,m,h}=f(n,m,h)$ is a function defined on the 3D lattice, we identify $\wt\cdot$, $\wh\cdot$ and $\wb\cdot$ with $\rT_p$, $\rT_q$ and 
$\rT_r$, i.e. 
\begin{align*}
 \wt f\doteq\rT_p f=f_{n+1,m,h}, \quad 
 \wh f\doteq\rT_q f=f_{n,m+1,h}, \quad 
 \wb f\doteq\rT_r f=f_{n,m,h+1},
\end{align*}
and similarly the undershifts denote the backward shifts, namely,
\begin{align*}
 \ut f\doteq\rT_{p}^{-1}f=f_{n-1,m,h}, \quad
 \uh f\doteq\rT_{q}^{-1}f=f_{n,m-1,h}, \quad 
 \ub f\doteq\rT_{r}^{-1}f=f_{n,m,h-1}.
\end{align*}

In order to consider a simple nontrivial solution to the closure relation \eqref{Closure}, we assume that $\fO_\fp$ is independent of $\fp$. 
The simplest choice is that $\fO_\fp=\bO$. At the same time, we require that $\fL_\fp$ and $\tfL_\fp$ linearly depend on the operators $\bLd$ 
and $\tbLd$ respectively. For nontrivial 3D lattice equations, a simple but nontrivial choice is 
\begin{align}\label{AKP:Operators}
 \fL_\fp=\fp+\bLd, \quad \tfL_\fp=\fp-\tbLd, \quad \fO_\fp=\bO, \quad \fp=p,q,r,
\end{align}
and as a result, the fundamental dynamical relations in \eqref{UDyn} in this case now turn out to be 
\begin{align}\label{AKP:UDyn}
 \wt\bU\cdot(p-\tbLd)=(p+\bLd)\cdot\bU-\wt\bU\cdot\bO\cdot\bU
\end{align}
together with the other two relations by the replacements $(\wt\cdot,p)\leftrightarrow(\wh\cdot,q)\leftrightarrow(\wb\cdot,r)$. The equation 
\eqref{AKP:UDyn} together with the analogues associated with the other two lattice directions are the basic relations that describe the evolutions 
of the infinite matrix $\bU$. Closed-from scalar lattice equations in certain entries of $\bU$ can be obtained from the relations. 
In order to find these lattice equations, we introduce 
\begin{align*}
 u=(\bU)_{0,0}, \quad v=1-(\bU\cdot\tbLd^{-1})_{0,0}, \quad w=1-(\bLd^{-1}\cdot\bU)_{0,0}, \quad 
 z=(\bLd^{-1}\cdot\bU\cdot\tbLd^{-1})_{0,0}-z_0, 
\end{align*}
where $z_0=\frac{n}{p}+\frac{m}{q}+\frac{h}{r}$. 
Since Equation \eqref{AKP:UDyn} is a relation for infinite matrices, we can consider $[\eqref{AKP:UDyn}]_{0,0}$, $[\eqref{AKP:UDyn}\cdot\bLd^{-1}]_{0,0}$, 
$[\bLd^{-1}\cdot\eqref{AKP:UDyn}]_{0,0}$ and $[\bLd^{-1}\cdot\eqref{AKP:UDyn}\cdot\tbLd^{-1}]_{0,0}$ respectively and obtain some dynamical relations 
involving the above quantities. By eliminating other redundant variables, the following nonlinear 3D lattice equations, namely the lattice KP, mKP and SKP 
equations, are derived: 
\begin{align}
 &(p-\wt u)(q-r+\bt u-\th u)+(q-\wh u)(r-p+\th u-\hb u)+(r-\wb u)(p-q-\hb u-\bt u)=0, \tag{lattice KP} \\
 &p\Big(\frac{\th v}{\wh v}-\frac{\bt v}{\wb v}\Big)+q\Big(\frac{\hb v}{\wb v}
 -\frac{\th v}{\wt v}\Big)+r\Big(\frac{\bt v}{\wt v}-\frac{\hb v}{\wh v}\Big)=0, \tag{lattice mKP} \\
 &p\Big(\frac{\wh w}{\th w}-\frac{\wb w}{\bt w}\Big)+q\Big(\frac{\wb w}{\hb w}
 -\frac{\wt w}{\th w}\Big)+r\Big(\frac{\wt w}{\bt w}-\frac{\wh w}{\hb w}\Big)=0, \tag{lattice mKP} \\
 &\frac{(\th z-\wh z)(\hb z-\wb z)(\bt z-\wt z)}{(\th z-\wt z)(\hb z-\wh z)(\bt z-\wb z)}=1. \tag{lattice SKP}
\end{align}
We note that the equations in $v$ and $w$ are both the lattice mKP equation and the continuum limit of either gives us the fully continuous 
mKP equation. If fact, we can transfer one to the other by introducing a simple transform $(n,m,h)\rightarrow(-n,-m,-h)$. In addition, 
if we consider $s=(\bLd^{-1}\cdot\bU-\bU\cdot\tbLd^{-1})_{0,0}$, another scalar closed-form equation can be found and it is given by 
\begin{align}\label{FWKP}
 &p\Bigg(\frac{\bt s-\wb s}{(2-\bt s)(2+\wb s)}-\frac{\th s-\wh s}{(2-\th s)(2+\wh s)}\Bigg) \nonumber \\
 &\qquad+q\Bigg(\frac{\th s-\wt s}{(2-\th s)(2+\wt s)}-\frac{\hb s-\wb s}{(2-\hb s)(2+\wb s)}\Bigg)
 +r\Bigg(\frac{\hb s-\wh s}{(2-\hb s)(2+\wh s)}-\frac{\bt s-\wt s}{(2-\bt s)(2+\wt s)}\Bigg)=0.
\end{align}
The lattice KP, mKP and SKP equations as well as Equation \eqref{FWKP} are connected with each other via Miura-type transforms. 
As a remark, we note that the equations in $v$, $w$, $z$ and $s$ can be obtained separately by choosing the same 
$\fL_\fp$ and $\tfL_\fp$ in \eqref{AKP:Operators} with slightly different $\fO_\fp$, namely, 
\begin{align}\label{AKP:ModOperator}
 \fO_\fp=\tbLd\cdot\bO, \quad \fO_\fp=\bO\cdot\bLd, \quad \fO_\fp=\tbLd\cdot\bO\cdot\bLd, \quad \fO_\fp=\tfrac{1}{2}(\bO\cdot\bLd-\tbLd\cdot\bO),
\end{align}
respectively. One can even consider more complex $\fO_\fp$ with the same $\fL_\fp$ and $\tfL_\fp$ and obtain other lattice equations, but they could only 
exist in multi-component form. In fact, all these lattice equations belong to the lattice AKP class because effectively they can all fit into \eqref{AKP:UDyn} 
and therefore we can conclude that \eqref{AKP:UDyn} together with its counterparts in the other lattice directions is the fundamental relation to understand 
the structure hidden behind the lattice AKP equation. 

There are some other important forms of the lattice AKP equation. For later convenience we introduce the following new variables: 
\begin{align*}
 V_a=1-\Big(\bU\cdot\frac{1}{a+\tbLd}\Big)_{0,0}, \quad W_a=1-\Big(\frac{1}{a+\bLd}\cdot\bU\Big)_{0,0}, \quad
 S_{a,b}=\Big(\frac{1}{a+\bLd}\cdot\bU\cdot\frac{1}{b+\tbLd}\Big)_{0,0}.
\end{align*}
These quantities contain the parameters $a$ and $b$ and therefore they are more general and can degenerate to the previous main variables via suitable 
limits in terms of $a$ and $b$. One of the most important quantities in the integrable systems theory is the $\tau$-function and it in the lattice AKP 
case can be defined as $\tau=\det(1+\bOa\cdot\bC)$. Here the determinant of the infinite matrix  can be worked out according to what we have explained in 
Section \ref{S:DLT}. In fact, following from \eqref{OaDyn} and \eqref{CDyn}, we have the dynamical relations for $\bOa$ and $\bC$ obeying 
\bse\label{AKP:SP}
\begin{align}
 &\bOa\cdot(p+\bLd)-(p-\tbLd)\cdot\bOa=\bO, \label{AKP:OaDyn} \\
 &\wt\bC\cdot(p-\tbLd)=(p+\bLd)\cdot\bC, \label{AKP:CDyn}
\end{align}
\ese
as well as the relations in terms of $(\wh\cdot,q)$ and $(\wb\cdot,r)$. These dynamical relations help us to find the evolutions of the $\tau$-function. 
Consider the dynamical evolution of the $\tau$-function, one can obtain 
\begin{align*}
 \wt\tau=\det[1+\bOa\cdot\bC+(p-\tbLd)^{-1}\cdot\bO\cdot\bC]=\tau\Big[1+\Big(\bU\cdot\frac{1}{p-\tbLd}\Big)_{0,0}\Big]
\end{align*}
and a similar relation for $\ut\tau$, where \eqref{WARank1} is used. Therefore we obtain the evolutions of the $\tau$-function as follows: 
\begin{align}\label{AKP:tauDyn}
 \frac{\wt\tau}{\tau}=V_{-p}, \quad \frac{\ut\tau}{\tau}=W_p.
\end{align}
where \eqref{AKP:SP} is used in the derivation. We note that similar relations also hold for $(\wh\cdot,q)$ and $(\wb\cdot,r)$. Now from \eqref{AKP:UDyn}, 
we can derive the dynamical relations for $W_a$ by considering $[\frac{1}{a+\bLd}\cdot\eqref{AKP:UDyn}]_{0,0}$, i.e. 
\begin{align}\label{AKP:WDyn}
 p(1-\wt W_a)-\Big(\frac{1}{a+\bLd}\cdot\wt\bU\cdot\tbLd\Big)_{0,0}=(p-a)(1-W_a)+\wt W_a u
\end{align}
as well as the similar relations for $(\wh\cdot,q)$ and $(\wb\cdot,r)$. Such relations give rise to the transform between $u$ and $\tau$ (if we set $a=p$): 
\begin{align}\label{AKP:u-tau}
 p-q+\wh u-\wt u=(p-q)\frac{\tau\th\tau}{\wh\tau\wt\tau}.
\end{align}
Similar relations for the other discrete shifts associated with the corresponding lattice parameters can be obtained in the same way. 
If we eliminate $u$ and replace $W_p$, $W_q$ and $W_r$ by the $\tau$-function, we end up with the discrete bilinear KP equation 
\begin{align}\label{dAKP}
 (p-q)\th\tau\wb\tau+(q-r)\hb\tau\wt\tau+(r-p)\bt\tau\wh\tau=0,
\end{align}
namely, the Hirota--Miwa equation, cf. \cite{Hir81} (see also \cite{Miw82}). Moreover, \eqref{AKP:UDyn} also gives rise to
\begin{align}\label{AKP:SDyn}
 1+(p-a)S_{a,b}-(p+b)\wt S_{a,b}=\wt W_aV_b, \quad (\wt\cdot,p)\leftrightarrow(\wh\cdot,q)\leftrightarrow(\wb\cdot,r),
\end{align}
if we consider $[\frac{1}{a+\bLd}\cdot\eqref{AKP:UDyn}\cdot\frac{1}{b+\tbLd}]_{0,0}$. Eliminating $W_a$ and $V_b$ in the relation then provides us with 
a more general 3D lattice equation which reads 
\begin{align}\label{AKP:NQC}
 \frac{[1+(p-a)\wh S_{a,b}-(p+b)\th S_{a,b}][1+(q-a)\wb S_{a,b}-(q+b)\hb S_{a,b}][1+(r-a)\wt S_{a,b}-(r+b)\bt S_{a,b}]}
 {[1+(p-a)\wb S_{a,b}-(p+b)\bt S_{a,b}][1+(q-a)\wt S_{a,b}-(q+b)\th S_{a,b}][1+(r-a)\wh S_{a,b}-(r+b)\hb S_{a,b}]}=1.
\end{align}
This is an extension of the lattice SKP equation which was given in \cite{NCWQ84}. The equation can be considered as the most general equation 
in the AKP class as suitable limits in terms of $a$ and $b$ of the equation lead to be the lattice KP, mKP and SKP equations as well as the 
Hirota--Miwa equation \eqref{dAKP}. In fact, the parameters $a$ and $b$ can be understood as the two extra lattice parameters associated 
with two auxiliary discrete variables, and both $a$ and $b$ as well as $p,q,r$ are on the same footing -- this is guaranteed by the MDC property. 
One can therefore find a connection between $S_{a,b}$ and the $\tau$-function, which is given by 
\begin{align}\label{AKP:S-tau}
 1-(a+b)S_{a,b}=\frac{\rT_a^{-1}\rT_{-b}\,\tau}{\tau}. 
\end{align}

Now we consider the linear problem of the lattice AKP equation. The linear equations \eqref{ukDyn} for the lattice AKP equation give us 
\begin{align}\label{AKP:ukDyn}
 \wt\bu_k=(p+\bLd)\cdot\bu_k-\wt\bU\cdot\bO\cdot\bu_k, \quad (\wt\cdot,p)\leftrightarrow(\wh\cdot,q)\leftrightarrow(\wb\cdot,r),
\end{align}
which leads to the following equation if we take the centre of \eqref{AKP:ukDyn} and set $\phi=(\bu_k)_0$: 
\begin{align}\label{AKP:Lax}
 \wt\phi-\wh\phi=(p-q+\wh u-\wt u)\phi=(p-q)\frac{\tau\th\tau}{\wh\tau\wt\tau}\phi,
\end{align}
and also the other two similar equations if we interchange $(\wt\cdot,p)\leftrightarrow(\wh\cdot,q)\leftrightarrow(\wb\cdot,r)$.
This is the Lax triplet of the Hirota--Miwa equation. In other words, the compatibility condition of \eqref{AKP:Lax} gives us \eqref{dAKP}.
One comment here is that if we eliminate the $\tau$-function in the Lax triplet, we can obtain a nonlinear form of the lattice AKP equation 
without the lattice parameters, namely, 
\begin{align}\label{AKP:phiEq}
 \frac{\th\phi-\hb\phi}{\wh\phi}+\frac{\hb\phi-\bt\phi}{\wb\phi}+\frac{\bt\phi-\th\phi}{\wt\phi}=0,
\end{align}
which is nothing but the lattice mKP equation of $v$ after a point transform. 

In the lattice KP, mKP equations $u$ and $v$ are both potential variables, we would like to note that there also exists a non-potential form of the 
lattice KP equation. The lattice KP equation can be written in an alternative form which reads 
\begin{align}\label{lpKP}
 \frac{(q-r+\wb u-\wh u)^{\wt{\ }}}{q-r+\wb u-\wh u}=\frac{(r-p+\wt u-\wb u)^{\wh{\ }}}{r-p+\wt u-\wb u}=\frac{(p-q+\wh u-\wt u)^{\wb{\ }}}{p-q+\wh u-\wt u}.
\end{align}
Owing to the above form of the lattice AKP equation, it is very natural for us to introduce the non-potential variables 
\begin{align}\label{AKP:NPVar}
 P=q-r+\wb u-\wh u, \quad Q=r-p+\wt u-\wb u, \quad R=p-q+\wh u-\wt u,
\end{align}
and consequently a coupled system of $P$, $Q$ and $R$ can be derived from \eqref{lpKP}, i.e. 
\begin{align}\label{AKP:NPSys}
 \frac{\wt P}{P}=\frac{\wh Q}{Q}=\frac{\wb R}{R}, \quad P+Q+R=0, \quad \wt P+\wh Q+\wb R=0,
\end{align}
where the first equation is Equation \eqref{lpKP} itself and the other two follow from the definitions of $P$, $Q$ and $R$ directly. The system 
\eqref{AKP:NPSys} can be treated as the lattice non-potential KP equation and it was proposed by Nimmo in \cite{Nim06}. Now for a scalar form, 
we first eliminate $R$ and obtain $\wt P/P=\wh Q/Q, \wb P+\wb Q=\wt P+\wh Q$, and then one can have the expression of $P$ in terms of $Q$ and 
consequently the expression of $R$ in terms of $Q$: 
\begin{align*}
 P=\frac{Q\wt Q(\wb Q(\th Q-\bt Q)+\hb Q(\wb Q-\wh Q))}{\wh Q(\wt Q\hb Q-\wb Q\th Q)}, \quad 
 R=\frac{Q\wb Q(\wt Q(\bt Q-\hb Q)+\th Q(\wh Q-\wb Q))}{\wh Q(\wt Q\hb Q-\wb Q\th Q)}.
\end{align*}
Now if we express everything in $Q$, the system turns out to be a 10-point scalar lattice equation (see Figure \ref{F:NPKP}) in the form of 
\begin{align}\label{AKP:NPEq}
 &\wt Q\wt{\wt Q}\wh{\wt Q}\thb Q\hb Q - \wt{\wt Q}\th Q\thb Q\wh Q\hb Q + \wt{\wt Q}\th Q\thb Q\hb Q\wb Q
 -\wt Q\wt{\wt Q}\th{\wt Q}\hb Q\bt Q - \wt Q\wt{\wt Q}\thb Q\hb Q\bt Q + \th{\wt Q}\th Q\wh Q\hb Q\bt Q\nonumber \\
 &+\wt{\wt Q}\th{\wt Q}\th Q\wb Q\bt Q - \th{\wt Q}{\th Q}^2\wb Q\bt Q - \th{\wt Q}\th Q\hb Q\wb Q\bt Q
 +\th{\wt Q}\th Q\wb Q{\bt Q}^2 + \wt Q\wt{\wt Q}\hb Q\bt Q\wt{\bt Q} - \wt{\wt Q}\th Q\wb Q\bt Q\wt{\bt Q}=0.
\end{align}
It is this quintic equation which we refer to as the non-potential lattice KP equation and this scalar form, to the best of the authors' knowledge, 
has not been given elsewhere. The Lax triplet of the equation can be obtained from \eqref{AKP:Lax} directly if we consider \eqref{AKP:NPVar} and the 
expressions of $P$ and $R$ in terms of $Q$ given above. We note that a different nonlinear form of the non-potential discrete AKP equation was given 
in \cite{GRPSW07} in which the equation is quartic but has 14 terms (compared to our quintic equation having 12 terms). 
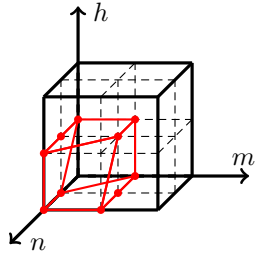
\begin{figure}
\centering
\begin{tikzpicture}[scale=0.15]
 \node at (-0.5,-3) {$n$};
 \node at (17.5,4.5) {$m$};
 \node at (5,17.5) {$h$};
 \draw[very thick,->] (3,3) -- (-3,-3);
 \draw[very thick,->] (3,3) -- (18,3);
 \draw[very thick,->] (3,3) -- (3,18);
 \coordinate (A1) at (0, 0);
 \coordinate (A2) at (10, 0);
 \coordinate (A3) at (13, 3);
 \coordinate (A4) at (3, 3);
 \coordinate (B1) at (0, 10);
 \coordinate (B2) at (10, 10);
 \coordinate (B3) at (13, 13);
 \coordinate (B4) at (3, 13);
 \draw[very thick] (A1) -- (A2);
 \draw[very thick] (A2) -- (A3);
 \draw[dashed,very thick] (A3) -- (A4);
 \draw[dashed,very thick] (A4) -- (A1);
 \draw[very thick] (B1) -- (B2);
 \draw[very thick] (B2) -- (B3);
 \draw[very thick] (B3) -- (B4);
 \draw[very thick] (B4) -- (B1);
 \draw[very thick] (A1) -- (B1);
 \draw[very thick] (A2) -- (B2);
 \draw[very thick] (A3) -- (B3);
 \draw[dashed,very thick] (A4) -- (B4);
 \foreach \x in {0,5,...,10}
 \foreach \y in {0,5,...,10}
  {
  \draw[densely dashed] (0+0.3*\y,\x+0.3*\y) -- (10+0.3*\y,\x+0.3*\y);
  \draw[densely dashed] (0+\x,0+\y) -- (3+\x,3+\y);
  \draw[densely dashed] (0+\x+0.3*\y,0+0.3*\y) -- (0+0.3*\y+\x,10+0.3*\y);
  }
 \fill [red] (0,0) circle (10pt);
 \fill [red] (1.5,1.5) circle (10pt);
 \fill [red] (6.5,6.5) circle (10pt);
 \fill [red] (8,8) circle (10pt);
 \fill [red] (1.5,6.5) circle (10pt);
 \fill [red] (6.5,1.5) circle (10pt);
 \fill [red] (0,5) circle (10pt);
 \fill [red] (3,8) circle (10pt);
 \fill [red] (8,3) circle (10pt);
 \fill [red] (5,0) circle (10pt);
 \draw[thick,red] (0,0) -- (1.5,1.5) -- (8,3) -- (6.5,1.5) -- (5,0) -- (0,0);
 \draw[thick,red] (0,0) -- (0,5) -- (6.5,6.5) -- (5,0) -- (0,0);
 \draw[thick,red] (0,0) -- (1.5,1.5) -- (3,8) -- (0,5) -- (0,0);
 \draw[thick,red] (0,5) -- (3,8) -- (8,8) -- (6.5,6.5) -- (0,5);
 \draw[thick,red] (5,0) -- (8,3) -- (8,8) -- (6.5,6.5) -- (5,0);
\end{tikzpicture}
\caption{The non-potential lattice AKP equation}
\label{F:NPKP}
\end{figure}

Before we move onto the lattice BKP and CKP equations, we would like to note that in the AKP class there exist several nonlinear integrable equations 
which are connected with each other through Miura-type transforms. In fact they are the same equation in different forms because the core part of the 
scheme, namely, the form of the operators $\fL_\fp$, $\tfL_\fp$ and $\fO_\fp$, has been fixed at very beginning. In fact, all these equations share the 
same bilinear structure, namely, the Hirota--Miwa equation. In order to understand and compare equations in different classes, we need a uniformising 
quantity to characterise the algebraic structure behind the 3D lattice equations. The $\tau$-function could be a very strong candidate as it only contains 
the key information of an integrable system, namely, the Cauchy kernel and the plane wave factors. Therefore in the future we only mean the equation in the 
$\tau$-function, i.e. the Hirota--Miwa equation \eqref{dAKP}, by the lattice AKP equation. While the other forms of the lattice AKP equation, namely, the 
equations in $u$, $v$, $w$, $S_{a,b}$ as well as Equation \eqref{FWKP} can all be expressed by the $\tau$-function. This is done by taking suitable limits 
on the parameters $a$ and $b$ in \eqref{AKP:S-tau}. In the next sections, we will mainly be concentrating on the $\tau$-function to understand the structure 
of the lattice BKP and CKP equations. Furthermore, we also note that the $\tau$-equation does not necessarily need to be bilinear although the terminology 
was from the bilinear theory.

\section{The lattice BKP equation}\label{S:BKP}
Now we try to find slightly complex solution to the closure relation \eqref{Closure} in order to obtain different classes of 3D integrable lattice equations. 
As we assumed that $\fO_\fp=\bO$ (the simplest case when $\fO_\fp$ is independent of $\fp$) and $\fL_\fp$ and $\tfL_\fp$  must linearly depend on $\bLd$ and 
$\tbLd$ in the lattice AKP class, we now suppose that they depend on $\bLd$ and $\tbLd$ fractionally linearly. One of solutions under the assumption can be 
 \begin{align}\label{BKP:Operators}
 \fL_\fp=\frac{\fp+\bLd}{\fp-\bLd}, \quad \tfL_\fp=\frac{\fp-\tbLd}{\fp+\tbLd}, \quad 
 \fO_\fp=\fp\frac{1}{\fp+\tbLd}\cdot(\bO\cdot\bLd-\tbLd\cdot\bO)\cdot\frac{1}{\fp-\bLd}, \quad \fp=p,q,r.
\end{align}
In order to understand the 3D lattice equation in this class from the point view of the $\tau$-function, we need the relations for $\bOa$ and $\bC$. 
According to \eqref{OaDyn} and \eqref{CDyn}, in this case the dynamical evolution for them must obey 
\bse\label{BKP:SP}
\begin{align}
 &\bOa\cdot\frac{p+\bLd}{p-\bLd}-\frac{p-\tbLd}{p+\tbLd}\cdot\bOa=p\frac{1}{p+\tbLd}\cdot(\bO\cdot\bLd-\tbLd\cdot\bO)\cdot\frac{1}{p-\bLd}, \label{BKP:OaDyn} \\
 &\wt\bC\cdot\frac{p-\tbLd}{p+\tbLd}=\frac{p+\bLd}{p-\bLd}\cdot\bC, \label{BKP:CDyn}
\end{align}
\ese
as well as the similar $(\wh\cdot,q)$- and $(\wb\cdot,r)$-relations. Equation \eqref{BKP:Operators} also tells us that the nonlinear evolutions of 
the infinite matrix $\bU$ in this class satisfy 
\begin{align}\label{BKP:UDyn}
 \wt\bU\cdot\frac{p-\tbLd}{p+\tbLd}=\frac{p+\bLd}{p-\bLd}\cdot\bU
 -p\,\wt\bU\cdot\frac{1}{p+\tbLd}\cdot(\bO\cdot\bLd-\tbLd\cdot\bO)\cdot\frac{1}{p-\bLd}\cdot\bU,
\end{align}
as well as similar relations for $(\wh\cdot,q)$ and $(\wb\cdot,r)$. These relations can be obtained from \eqref{UDyn}. Now one can try to select 
certain entries in the matrix $\bU$ following the idea in the lattice AKP equation and expect to derive some closed-form equations from the framework. 
However, due to $\fO_\fp$ now depending on the lattice parameters $p,q,r$, it is not possible to derive closed-form scalar lattice equations from 
the relations without any extra conditions and a multi-component form must be involved. In order to get a scalar equation, we impose an extra condition 
on the relations and it is actually done through imposing a constraint on the measure in the DLT which is given by $\rd\zeta(l,l')=-\rd \zeta(l',l)$. 
The reason why we impose the extra condition is that the $\bOa$ in this case is antisymmetric if one follows \eqref{OaDyn} and \eqref{Kernel} and we 
want the measure to match this property. Under this constraint, one can figure out from \eqref{C} and \eqref{U} that the infinite matrices $\bC$ and 
$\bU$ satisfy the antisymmetry property $\tbC=-\bC, \tbU=-\bU$. Consequently we have $(\bLd^i\cdot\bU\cdot\tbLd^i)_{0,0}=0$ for arbitrary $i\in\mathbb{Z}$. 
For simplicity in the future, we here introduce some new variables as follows: 
\begin{align*}
 &V_a=1-\Big(\bU\cdot\frac{a}{a-\tbLd}\Big)_{0,0}=1+\Big(\frac{a}{a-\bLd}\cdot\bU\Big)_{0,0}=W_a, \\
 &S_{a,b}=\Big(\frac{a}{a-\bLd}\cdot\bU\cdot\frac{b}{b-\tbLd}\Big)_{0,0}=-S_{b,a},
\end{align*}
where $V_a=W_a$ and $S_{a,b}=-S_{b,a}$ due to the antisymmetry property of $\bU$. Focusing on these new variables and consider the 
centre of the dynamical evolutions including $[\eqref{BKP:UDyn}]_{0,0}$, $[\frac{a}{a-\bLd}\cdot\eqref{BKP:UDyn}]_{0,0}$ and also 
$[\frac{a}{a-\bLd}\cdot\eqref{BKP:UDyn}\cdot\frac{b}{b-\tbLd}]_{0,0}$, we derive the following relations as a consequence: 
\bse 
\begin{align}
 &V_p\wt V_{-p}=1, \label{BKP:uDyn} \\
 &1+2V_p\wt S_{-a,-p}=\frac{a-p}{a+p}(V_p-V_{-a})+V_p\wt V_{-a} \label{BKP:VDyn}
\end{align}
as well as the dynamical relation for the variable $S_{a,b}$: 
\begin{align}
 &\frac{b+p}{b-p}\wt S_{-a,-b}+\frac{p-a}{p+a}S_{-a,-b}-\frac{2b}{b-p}\wt S_{-a,-p}+\frac{2a}{p+a}S_{p,-b} \nonumber \\
 &\qquad+(1-V_{-b})\wt S_{-a,-p}-(1-\wt V_{-a}+2\wt S_{-a,-b})S_{p,-b}=0, \label{BKP:SDyn}
\end{align}
\ese
together with the dynamical relations along the other two directions, i.e. $(\wh\cdot,q)$ and $(\wb\cdot,r)$. These relations will play the core ingredients 
for constructing a closed-form lattice equation later. We now introduce the $\tau$-function and it is defined by $\tau^2=\det(1+\bOa\cdot\bC)$. 
Consider the dynamical evolution of the $\tau$-function, we have 
\begin{align*}
 \wt\tau^2={}&\det(1+\bOa\cdot\wt\bC)
 =\det\Big(1+\bOa\cdot\bC+p\,\frac{1}{p-\tbLd}\cdot(\bO\cdot\bLd-\tbLd\cdot\bO)
 \cdot\frac{1}{p-\bLd}\cdot\bC\Big) \nonumber \\
 ={}&\tau^2\det
 \left(
 \begin{array}{cc}
 1+p(\frac{\bLd}{p-\bLd}\cdot\bU\cdot\frac{1}{p-\tbLd})_{0,0} & -p(\frac{\bLd}{p-\bLd}\cdot\bU\cdot\frac{\tbLd}{p-\tbLd})_{0,0}\\
 p(\frac{1}{p-\bLd}\cdot\bU\cdot\frac{1}{p-\tbLd})_{0,0} & 1-p(\frac{1}{p-\bLd}\cdot\bU\cdot\frac{\tbLd}{p-\tbLd})_{0,0}
 \end{array}
 \right),
\end{align*}
in which we have used the rank 2 Weinstein--Aronszajn formula. The similar relations also hold for $(\wh\cdot,q)$ and $(\wb\cdot,r)$ 
as well as the undershifts $(\ut\cdot,-p)$, $(\uh\cdot,-q)$ and $(\ub\cdot,-r)$. Taking \eqref{BKP:uDyn} into consideration, without 
losing generality, we have the dynamical relations for the $\tau$-function 
\begin{align}\label{BKP:tauDyn}
 V_p=\frac{\wt\tau}{\tau}, \quad V_{-p}=\frac{\ut\tau}{\tau}, \quad
 V_q=\frac{\wh\tau}{\tau}, \quad V_{-q}=\frac{\uh\tau}{\tau}, \quad
 V_r=\frac{\wb\tau}{\tau}, \quad V_{-r}=\frac{\ub\tau}{\tau}.
\end{align}
Now if we set $a=q$ in \eqref{BKP:VDyn}, the $V$-variables can be replaced by the $\tau$-function using \eqref{BKP:tauDyn}, and therefore we can express 
$S_{-q,-p}$ by $\tau$ in the following formula:
\begin{align}\label{BKP:S-tau}
 S_{-q,-p}=\frac{1}{2}\Big[\frac{\uh\tau-\ut\tau}{\tau}+\frac{q-p}{q+p}\Big(1-\frac{\uth\tau}{\tau}\Big)\Big].
\end{align}
Similar relations in terms of the other lattice parameters and discrete shifts can be obtained in a similar way. Finally, setting $a=q$ and $b=r$ in 
\eqref{BKP:SDyn} and replacing everything by the $\tau$-function, we end up with a closed-form 3D lattice equation in the $\tau$-function, 
which takes the form of 
\begin{align}\label{dBKP}
 &(p-q)(q-r)(r-p)\tau\thb\tau+(p+q)(p+r)(q-r)\wt\tau\hb\tau \nonumber \\
 &\qquad+(r+p)(r+q)(p-q)\wb\tau\th\tau+(q+r)(q+p)(r-p)\wh\tau\bt\tau=0.
\end{align}
This is the lattice BKP equation which is also known as the Miwa equation \cite{Miw82}. The difference between the Miwa equation and the Hirota--Miwa 
equation is that the former has an extra term (the fourth term). In fact if one applies some formal limits, the Hirota--Miwa equation can be recovered 
from the Miwa equation. 

Next, we consider the linear problem of the lattice BKP equation. The linear equations \eqref{ukDyn} in the lattice BKP case are given by 
\begin{align}\label{BKP:ukDyn}
 \wt\bu_k=\frac{p+\bLd}{p-\bLd}\cdot\bu_k-p\,\wt\bU\cdot\frac{1}{p+\tbLd}\cdot(\bO\cdot\bLd-\tbLd\cdot\bO)\cdot\frac{1}{p-\bLd}\cdot\bu_k
\end{align}
together with its $(\wh\cdot,q)$ and $(\wb\cdot,r)$ counterparts, as $\fL_\fp$, $\tfL_\fp$ and $\fO_\fp$ follow from \eqref{BKP:Operators}.
By introducing $\phi=(\bu_k)_0, \psi_a=(\frac{a}{a-\bLd}\cdot\bu_k)_0$, we therefore from \eqref{BKP:ukDyn} have 
\begin{align}
 \wt\phi=\wt V_{-p}(\psi_p-\phi), \quad
 \wt\psi_{-a}=\frac{p-a}{p+a}(\psi_{-a}-\psi_p)+(\wt V_{-a}-2\wt S_{-a,-p})\psi_p+\wt S_{-a,-p}\phi \label{BKP:psiDyn}
\end{align}
together with the similar equations associated with $(\wh\cdot,q)$ and $(\wb\cdot,r)$. And now we eliminate $\psi$ by setting $a=p$ and as a result
we can derive the Lax triplet of the lattice BKP equation 
\begin{align}\label{BKP:Lax}
 \th\phi-\phi=\Big(\frac{p+q}{p-q}\Big)\frac{\wh\tau\wt\tau}{\tau\th\tau}(\wh\phi-\wt\phi), \quad
 (\wt\cdot,p)\leftrightarrow(\wh\cdot,q)\leftrightarrow(\wb\cdot,r).
\end{align}
Similarly, a nonlinear form of the BKP equation can be obtained if we eliminate the $\tau$-function in the Lax triplet and it takes the form of 
\begin{align}\label{BKP:phiEq}
 \frac{(\hb\phi-\bt\phi)(\th\phi-\phi)}{(\bt\phi-\th\phi)(\phi-\hb\phi)}=\frac{(\wt\phi-\wh\phi)(\wb\phi-\thb\phi)}{(\wh\phi-\wb\phi)(\thb\phi-\wt\phi)},
\end{align}
where two cross-ratios are involved in the explicit form.

\section{The lattice CKP equation}\label{S:CKP}

Under the assumption in Section \ref{S:BKP}, namely, the operators $\fL_\fp$, $\tfL_\fp$ and $\fO_\fp$ depend on $\bLd$ and $\tbLd$ fractionally linearly, 
there is also another solution to the closure relation \eqref{Closure}: 
\begin{align}\label{CKP:Operators}
 \fL_\fp=\frac{\fp+\bLd}{\fp-\bLd}, \quad \tfL_\fp=\frac{\fp-\tbLd}{\fp+\tbLd}, \quad 
 \fO_\fp=2\fp\frac{1}{\fp+\tbLd}\cdot\bO\cdot\frac{1}{\fp-\bLd}, \quad \fp=p,q,r.
\end{align}
Again, once we have the explicit expressions for the key information $\fL_\fp$, $\tfL_\fp$ and $\fO_\fp$, the general formula \eqref{UDyn} for the 
infinite matrix $\bU$ turns out to be 
\begin{align}\label{CKP:UDyn}
 \wt\bU\cdot\frac{p-\tbLd}{p+\tbLd}=\frac{p+\bLd}{p-\bLd}\cdot\bU
 -2p\,\wt\bU\cdot\frac{1}{p+\tbLd}\cdot\bO\cdot\frac{1}{p-\bLd}\cdot\bU
\end{align}
together with the analogues after the replacements $(\wt\cdot,p)\rightarrow(\wh\cdot,q)$ as well as $(\wt\cdot,p)\rightarrow(\wb\cdot,r)$. 
We note that it is very similar to the lattice BKP equation that a closed-form scalar 3D lattice equation might not exist from \eqref{CKP:UDyn} 
without extra conditions due to the complexity of $\fO_\fp$. Now we consider the general dynamical relations for $\bOa$ and $\bC$ in this case. 
According the forms of $\fL_\fp$, $\tfL_\fp$ and $\fO_\fp$, we have from \eqref{OaDyn} and \eqref{CDyn} 
\bse\label{CKP:SP}
\begin{align}
 &\bOa\cdot\frac{p+\bLd}{p-\bLd}-\frac{p-\tbLd}{p+\tbLd}\cdot\bOa=2p\frac{1}{p+\tbLd}\cdot\bO\cdot\frac{1}{p-\bLd}, \label{CKP:OaDyn} \\
 &\wt\bC\cdot\frac{p-\tbLd}{p+\tbLd}=\frac{p+\bLd}{p-\bLd}\cdot\bC \label{CKP:CDyn}
\end{align}
\ese
as well as the relations by replacing $(\wt\cdot,p)$ by $(\wh\cdot,q)$ and $(\wb\cdot,r)$. The important thing is that the $\bOa$ in \eqref{CKP:OaDyn} 
turns out to be a symmetric kernel if one takes \eqref{Kernel} into consideration. Motivated by the antisymmetry condition on the measure in the lattice 
BKP, we impose a symmetry condition on the measure in the DLT in this case, i.e. $\rd\zeta(l,l')=\rd\zeta(l',l)$. In other words, we let the measure match 
the property of the Cauchy kernel. As a result, we have the symmetry property for the infinite matrices $\bC$ and $\bU$ if we follow the definitions \eqref{C} 
and \eqref{U}, namely, $\tbC=\bC, \tbU=\bU$. Now we introduce some new variables in which some new parameters are involved: 
\begin{align*}
 V_a=1-\Big(\bU\cdot\frac{1}{a+\tbLd}\Big)_{0,0}=1-\Big(\frac{1}{a+\bLd}\cdot\bU\Big)_{0,0}=W_a, \quad 
 S_{a,b}=\Big(\frac{1}{a+\bLd}\cdot\bU\cdot\frac{1}{b+\tbLd}\Big)_{0,0}=S_{b,a},
\end{align*}
where $V_a=W_a$ and $S_{a,b}=S_{b,a}$ hold because of the symmetry property of $\bU$. Now the dynamical evolution of $\bU$ \eqref{CKP:UDyn} can 
give rise to the evolutions for the introduced new variables as follows if we consider $[\frac{1}{a+\bLd}\cdot\eqref{CKP:UDyn}]_{0,0}$ and 
$[\frac{1}{a+\bLd}\cdot\eqref{CKP:UDyn}\cdot\frac{1}{b+\tbLd}]_{0,0}$: 
\bse
\begin{align}
 &\wt V_a+\frac{p-a}{p+a}V_a=\Big(\frac{2p}{p+a}-\wt S_{a,p}\Big)V_{-p}, \label{CKP:VDyn} \\
 &[1-(p+a)\wt S_{a,p}][1+(p-b)S_{-p,b}]+\frac{(p+a)(p+b)}{2p}\wt S_{a,b}-\frac{(p-a)(p-b)}{2p}S_{a,b}=1 \label{CKP:SDyn}
\end{align}
\ese
as well as the analogues in terms of $(\wh\cdot,q)$ and $(\wb\cdot,r)$. The $\tau$-function in this class is defined by $\tau=\det(1+\bOa\cdot\bC)$. 
Since we know the dynamical evolutions of $\bOa$ and $\bC$ from \eqref{CKP:SP}, with the help of them the dynamical evolution of $\tau$-function gives us 
\begin{align*}
 \wt\tau&=\det\Big(1+\bOa\cdot\bC+2p\frac{1}{p-\tbLd}\cdot\bO\cdot\frac{1}{p-\bLd}\cdot\bC\Big)
 =\tau\Big[1+2p\Big(\frac{1}{p-\bLd}\cdot\bU\cdot\frac{1}{p-\tbLd}\Big)_{0,0}\Big],
\end{align*}
where the computation is very similar to the case in the lattice AKP equation. Similarly, we can also derive the evolution of $\tau$-function in 
terms of the undershifts. Therefore the expressions for the $\tau$-function are given by 
\begin{align}\label{CKP:tauDyn}
 \frac{\wt\tau}{\tau}=1+2pS_{-p,-p}, \quad \frac{\ut\tau}{\tau}=1-2pS_{p,p}
\end{align}
together with the counterparts of the shifts $\wh\cdot$ and $\wb\cdot$ associated with their corresponding lattice parameters. Furthermore, 
if we set $a=p$ in \eqref{CKP:VDyn} and make use of the above relations, the dynamical relations of $V$-variables in terms of $\tau$ can easily 
be obtained, which are 
\begin{align}\label{CKP:V-tau}
 \frac{\wt V_p}{V_{-p}}=\frac{\tau}{\wt\tau}, \quad \frac{\ut V_{-p}}{V_p}=\frac{\tau}{\ut\tau},
\end{align}
and the similar relations also hold via the replacements $(\wt\cdot,p)\leftrightarrow(\wh\cdot,q)\leftrightarrow(\wb\cdot,r)$. 
In order to find a closed-form equation in the $\tau$-function, we first set $a=p$ and $b=q$ in \eqref{CKP:SDyn} respectively and this gives us 
\begin{align*}
 \frac{1-(p+b)\wt S_{p,b}}{1+(p-b)S_{-p,b}}=\frac{\tau}{\wt\tau}, \quad \frac{1+(p-a)S_{a,-p}}{1-(p+a)\wt S_{a,p}}=\frac{\wt\tau}{\tau}, \quad 
 (\wt\cdot,p)\leftrightarrow(\wh\cdot,q)\leftrightarrow(\wb\cdot,r).
\end{align*}
Now if we set $a=b=q$ in \eqref{CKP:SDyn} and make use of the above obtained relations, an expression of the $S$-variable in terms of 
the $\tau$-function can be obtained as follow: 
\begin{align}\label{CKP:S-tau}
 [1-(p+q)S_{p,q}]^2=\frac{(p+q)^2\ut\tau\uh\tau-(p-q)^2\tau\uth\tau}{4pq\tau^2}.
\end{align}
The expressions for $S_{q,r}$ and $S_{r,p}$ by the $\tau$-function as well as the $S$-quantity associated with $-p,-q,-r$ can be derived in a 
similar way. Equation \eqref{CKP:S-tau} is the key relation in deriving a closed-form 3D lattice equation. In fact, if we set $a=q$ and $b=p$ in 
\eqref{CKP:SDyn} and express all the variables in the equation by the $\tau$-function, it turns out to be the lattice CKP equation 
\begin{align}\label{dCKP}
 &[(p-q)^2(q-r)^2(r-p)^2\tau\thb\tau+(p+q)^2(p+r)^2(q-r)^2\wt\tau\hb\tau \nonumber \\
 &\quad-(q+r)^2(q+p)^2(r-p)^2\wh\tau\bt\tau-(r+p)^2(r+q)^2(p-q)^2\wb\tau\th\tau]^2 \nonumber \\
 &\qquad-4(p^2-q^2)^2(r^2-p^2)^2[(q+r)^2\th\tau\bt\tau-(q-r)^2\wt\tau\thb\tau][(q+r)^2\wh\tau\wb\tau-(q-r)^2\tau\hb\tau]=0.
\end{align}
This is a new parametrisation for the lattice CKP equation in contrast to the existed form in \cite{Kas96,Sch03}. This version of the lattice CKP equation, 
i.e. \eqref{dCKP}, can provide soliton solutions as the lattice parameters $p,q,r$ are introduced. Furthermore, we note that the left hand side of the 
equation takes the form of Cayley's $2\times2\times2$ hyperdeterminant. 

Now we consider the linear problem of the lattice CKP equation. Equation \eqref{ukDyn} in the case of the lattice CKP equation is in the form of 
\begin{align}\label{CKP:ukDyn}
 \wt\bu_k=\frac{p+\bLd}{p-\bLd}\cdot\bu_k-2p\,\wt\bU\cdot\frac{1}{p+\tbLd}\cdot\bO\cdot\frac{1}{p-\bLd}\cdot\bu_k, \quad
 (\wt\cdot,p)\leftrightarrow(\wh\cdot,q)\leftrightarrow(\wb\cdot,r).
\end{align}
The linear equation \eqref{CKP:ukDyn} provides us with 
\begin{align*}
 \psi_{-p}=\frac{1}{2p\wt V_p}(\wt\phi+\phi), \quad
 \wt\psi_a=\frac{p-a}{p+a}\psi_a-\frac{2p}{p+a}[1-(p+a)\wt S_{a,p}]\psi_{-p}
\end{align*}
together with its $(\wh\cdot,q)$ and $(\wb\cdot,r)$ counterparts, in which $\phi=(\bu_k)_0, \psi_a=(\frac{1}{a+\bLd}\cdot\bu_k)_0$. 
Now if we set $a=-q$ and eliminate the $\psi$-variables in the linear problem, we can obtain the Lax triplet of the lattice CKP equation, which is 
\begin{align}\label{CKP:Lax}
 \th\phi+\wt\phi=\frac{p+q}{p-q}\frac{\th V_q}{\wh V_q}(\phi+\wh\phi)+\frac{2q}{p-q}\frac{\th V_q}{\wt V_p}[1-(p-q)\wt S_{p,-q}](\phi+\wt\phi)
\end{align}
together with its counterparts associated with the other two directions, where the $V$- and $S$-variables can be expressed in terms of the 
$\tau$-function via Equations \eqref{CKP:V-tau} and \eqref{CKP:S-tau} respectively.

\section{Soliton solution structure}\label{S:Soliton}

We now discuss how explicit soliton solutions for the lattice AKP, BKP and CKP equations arise from the DL constructed in the previous sections. 
Actually, these soliton solutions structure arise by restricting the double integral (in the integral equation) on a domain which contains a finite 
number of singular points. In other words, a measure that can introduce a finite number of singularities is required. As we noted previously, a 3D 
integrable lattice equation may have different forms but the solution structure hidden behind them is always the same. For convenience, we only study 
soliton solution to the $\tau$-equations, i.e. the Hirota--Miwa (AKP) equation \eqref{dAKP}, the Miwa (BKP) equation \eqref{dBKP} and the Kashaev (CKP) 
equation \eqref{dCKP}. Soliton solutions to the other forms can easily be recovered via the associated discrete differential transforms. 

In the soliton reduction, we need the following finite matrices: an $N\times N'$ constant matrix $\bA$ with entries $A_{i,j}$, and a generalised Cauchy matrix 
\begin{align}\label{CauchyMat}
 \bM=(M_{j,i})_{N'\times N}, \quad M_{j,i}=\sigma_{k'_j}\Oa_{j,i}\rho_{k_i},
\end{align}
where one can use \eqref{PWF} and \eqref{Kernel} to figure out the plane wave factors $\rho_{k_i}$ and $\sigma_{k'_j}$ and the Cauchy kernel 
$\Oa_{j,i}=\Oa_{k_i,k'_j}$. Next, we impose the following condition on the measure $\rd \zeta(l,l')$: 
\begin{align}\label{Singular}
 \rd\zeta(l,l')=\sum_{i=1}^{N}\sum_{j=1}^{N'}A_{i,j}\delta(l-k_i)\delta(l'-k'_j)\rd l\rd l'.
\end{align}
In fact, $k_i$ and $k'_j$ are the singular points in the domain $D$ for soliton solutions. The condition of the measure \eqref{Singular} immediately 
gives rise to the degeneration of $\bC$ and it takes the form of 
\begin{align}\label{CFiniteSum}
 \bC=\sum_{i=1}^{N}\sum_{j=1}^{N'}A_{i,j}\rho_{k_i}\bc_{k_i}\tbc_{k'_j}\sigma_{k'_j}.
\end{align}
The degeneration helps us to get rid of the the nonlocality (i.e. the double integral) in the problem and brings us a finite summation which 
leads to soliton solutions of the 3D lattice equations in the form of a finite matrix. We conclude it as the following important statement: 
\begin{align}\label{NSoliton}
 \det(1+\bOa\cdot\bC)=\det(\bI_{N'\times N'}+\bM_{N'\times N}\bA_{N\times N'})=\det(\bI_{N\times N}+\bA_{N\times N'}\bM_{N'\times N}).
\end{align}
Either equality in Equation \eqref{NSoliton} provides us with a general formula for soliton solutions to the 3D lattice equations. In practice, once 
the operators $\fL_\cdot$, $\tfL_\cdot$ and $\fO_\cdot$ are given, one can recover the plane wave factors $\rho_k$ and $\sigma_{k'}$ as well as the 
Cauchy kernel $\Oa_{k,k'}$ and compose the generalised Cauchy matrix $\bM$. Meanwhile, the restriction on the measure $\rd \zeta(l,l')$ may require that 
the constant matrix $\bA$ satisfy certain conditions. As a result, soliton solutions can be constructed. In the following, we follow from this idea and 
give the $N$-soliton solutions to the lattice AKP, BKP and CKP equations one by one. 

\paragraph{Lattice AKP equation.}
The operators following from \eqref{AKP:Operators} in the lattice AKP equations result in the plane wave factors and the Cauchy kernel as follows: 
\begin{align}\label{AKP:PWF}
 \rho_{k_i}=(p+k_i)^n(q+k_i)^m(r+k_i)^h, \quad \sigma_{k'_j}=(p-k'_j)^{-n}(q-k'_i)^{-m}(r-k'_j)^{-h}, \quad \Oa_{j,i}=\frac{1}{k_i+k'_j}.
\end{align}
As there is no restriction on the measure $\rd \zeta(l,l')$ in the lattice AKP case, the matrix $\bA$ can be arbitrary (but non-degenerate). Therefore 
the $(N,N')$-soliton solution to the Hirota--Miwa equation \eqref{dAKP} is 
\begin{align}\label{AKP:Sol}
 \tau=\det(\bI+\bA\bM), \quad M_{j,i}=\frac{\rho_{k_i}\sigma_{k'_j}}{k_i+k'_j}, \quad i=1,2,\cdots,N,\quad j=1,2,\cdots,N'.
\end{align}

\paragraph{Lattice BKP equation.}
The plane wave factors and the Cauchy kernel for the lattice BKP equation are given by 
\begin{align}\label{BKP:PWF}
 \rho_{k_i}=\Big(\frac{p+k_i}{p-k_i}\Big)^n\Big(\frac{q+k_i}{q-k_i}\Big)^m\Big(\frac{r+k_i}{r-k_i}\Big)^h, \quad 
 \sigma_{k'_j}=\rho_{k'_j}, \quad 
 \Oa_{j,i}=\frac{1}{2}\frac{k_i-k'_j}{k_i+k'_j},
\end{align}
which follow from \eqref{BKP:Operators}. In addition, in the lattice BKP equation we impose the antisymmetry condition 
$\rd\zeta(l,l')=-\rd \zeta(l',l)$ on the measure and this leads to $A_{i,j}=-A_{j,i}$ in the matrix $\bA$, i.e. 
\begin{align}\label{BKP:Sigular}
  \rd\zeta(l,l')=\sum_{i,j=1}^{2N}A_{i,j}\delta(l-k_i)\delta(l'-k'_j)\rd l\rd l', \quad A_{i,j}=-A_{j,i}.
\end{align}
As a result, the $N$-soliton solution to the Miwa equation \eqref{dBKP} is determined by 
\begin{align}\label{BKP:Sol}
 \tau^2=\det(\bI+\bA\bM), \quad M_{j,i}=\rho_{k_i}\frac{1}{2}\frac{k_i-k'_j}{k_i+k'_j}\sigma_{k'_j}, \quad A_{i,j}=-A_{j,i}, \quad i,j=1,2,\cdots,2N.
\end{align}
As the matrices $\bA$ and $\bM$ are both antisymmetric it can be shown that the determinant $\det(\bI+\bA\bM)$ must be a perfect square. 
In other words, the $\tau$-function itself can be expressed by a Pfaffian. 

\paragraph{Lattice CKP equation.}
Now the plane wave factors and the Cauchy kernel that follow from \eqref{CKP:Operators} in the lattice CKP equation are given by 
\begin{align}\label{CKP:PWF}
 \rho_{k_i}=\Big(\frac{p+k_i}{p-k_i}\Big)^n\Big(\frac{q+k_i}{q-k_i}\Big)^m\Big(\frac{r+k_i}{r-k_i}\Big)^h, \quad 
 \sigma_{k'_j}=\rho_{k'_j}, \quad 
 \Oa_{j,i}=\frac{1}{k_i+k'_j}.
\end{align}
Due to the symmetry property of the measure $\rd \zeta(l,l')=\rd \zeta(l',l)$, one has to set $N'=N$ and impose the symmetry condition 
on $\bA$, and therefore
\begin{align}\label{CKP:Sigular}
  \rd\zeta(l,l')=\sum_{i,j=1}^{N}A_{i,j}\delta(l-k_i)\delta(l'-k'_j)\rd l\rd l', \quad A_{i,j}=A_{j,i}.
\end{align}
Clearly the $N$-soliton solution to the Kashaev equation \eqref{dCKP} can then be written in the form of 
\begin{align}\label{CKP:Sol}
  \tau=\det(\bI+\bA\bM), \quad M_{j,i}=\frac{\rho_{k_i}\sigma_{k'_j}}{k_i+k'_j}, \quad A_{i,j}=A_{j,i}, \quad i,j=1,2,\cdots,N.
\end{align}

The paper only considers some simple choices of $\fL_\fp$, $\tfL_\fp$ and $\fO_\fp$ as solutions to the closure relations \eqref{Closure} such as 
\eqref{AKP:Operators}, \eqref{BKP:Operators} and \eqref{CKP:Operators}. Other choices also exist, for example, one can refer to \eqref{AKP:ModOperator} 
and the relevant results, however, the obtained 3D lattice equations still fit into the scheme of the discrete AKP class as this has been shown in 
Section \ref{S:AKP}, which implies that the choice for the operators is not unique for a certain class of 3D integrable lattice equations. This is because 
the operators are deeply related to the Cauchy kernels and the plane wave factors which are not invariants in the integrable systems theory. In other words, 
solutions to the closure relation \eqref{Closure} do not provide a classification of 3D discrete integrable systems. In the paper, we choose the simplest 
operators generating the master discrete AKP, BKP and CKP equations. Therefore, a full classification of 3D integrable lattice equations within the DL 
framework, remains a problem.

\paragraph{Acknowledgements.} 
WF was supported by a Leeds International Research Scholarship (LIRS) as well as a small grant from the School of Mathematics. 
FWN was partially supported by EPSRC (Ref. EP/I038683/1).

\end{document}